\documentclass[10pt, journal]{IEEEtran}
\ifCLASSINFOpdf
\else
\fi
\usepackage{lipsum}
\usepackage{graphicx}
\usepackage{float}
\usepackage{cite}
\usepackage{epstopdf}
\usepackage{multirow}
\usepackage{amssymb}
\usepackage{amsmath}
\usepackage{amsthm}
\usepackage{mathtools, cases}
\usepackage{chngcntr}
\usepackage{amsmath,etoolbox}% http://ctan.org/pkg/{amsmath,etoolbox}
\usepackage{amssymb}
\usepackage[ruled]{algorithm2e}
\usepackage{url}
\usepackage{subcaption}
\usepackage[font=small]{caption}
\usepackage{cleveref}
\usepackage{color}
\usepackage{nopageno}
\usepackage[normalem]{ulem}
\usepackage[utf8]{inputenc}
\usepackage[english]{babel}
\newtheorem{theorem}{Theorem}
\newtheorem{lemma}{Lemma}
\newtheorem{corollary}{Corollary}
\makeatletter
\renewcommand\footnoterule{%
  \kern-3\p@
  \hrule\@width.4\columnwidth
  \kern2.6\p@}
  \makeatother

\newcommand{\E}{\mathbb{E}}
\newcommand{\nt}{\tilde{n}}
\newcommand{\at}{\tilde{a}}

\newcommand{\setC}{\mathcal{C}}
\newcommand{\Rt}{\tilde{R}}

\newcommand{\bt}{\tilde{b}}

\newcommand{\qt}{\tilde{q}}
\begin{document}
\title{The Impact of Mobile Blockers on Millimeter Wave Cellular Systems}
%\{(Invited Paper)}
%\author{Ish Kumar Jain, Rajeev Kumar, and Shivendra Panwar\\
%Department of Electrical and Computer Engineering, Tandon School of Engineering, NYU, NY 11201, USA\vspace{-10mm}
% \\\textit{(Invited Paper)}\vspace{-3mm}
%}
\author{Ish~Kumar~Jain,~\IEEEmembership{Student~Member,~IEEE,}
Rajeev~Kumar,~\IEEEmembership{Student~Member,~IEEE,} and
Shivendra~Panwar,~\IEEEmembership{Fellow,~IEEE}
\thanks{This work was supported in part by the U.S. National Science Foundation  under  Grant 1527750, NYU Wireless, and by the NY State Center for Advanced Technology in Telecommunications (CATT).}
\thanks{Ish Kumar Jain was with the Department of Electrical and Computer Engineering, Tandon School of Engineering, New York University, Brooklyn, NY 11210, USA (e-mail: ishjain@nyu.edu). Currently, he is with Department of the Electrical and Computer Engineering at the University of California, San Diego, CA 99093, USA (e-mail: ikjain@ucsd.edu)}
\thanks{Rajeev Kumar and Shivendra Panwar are with the Department of Electrical and Computer Engineering, Tandon School of Engineering, New York University, Brooklyn, NY 11201, USA (e-mail: rajeevkr@nyu.edu; panwar@nyu.edu).}% 
\thanks{This paper earlier published on arxiv with a different name as: \textit{``Can Millimeter Wave Cellular Systems provide High Reliability and Low Latency? An analysis of the impact of Mobile Blockers.''} The paper has been accepted at IEEE Journal in Selected Area of Communication and published as: I. K. Jain, R. Kumar and S. Panwar, \emph{``The Impact of Mobile Blockers on Millimeter Wave Cellular Systems,''} in IEEE Journal on Selected Areas in Communications. This version of paper also corrects Figure 3 and Figure 11(b) of the published JSAC version.}
\vspace{-10mm}}
\maketitle
\thispagestyle{empty}
\begin{abstract}
Millimeter Wave (mmWave) communication systems can provide high data rates, but the system performance may degrade significantly due to interruptions by mobile blockers such as humans or vehicles. A high frequency of interruptions and lengthy blockage durations will degrade the quality of the user's experience. A promising solution is to employ the macrodiversity of Base Stations (BSs), where the User Equipment (UE) can handover to other available BSs if the current serving BS gets blocked. However, an analytical model to evaluate the system performance of dynamic blockage events in this setting is unknown. In this paper, we develop a Line of Sight (LOS) dynamic blockage model and evaluate the probability, duration, and frequency of blockage events considering all the links to the UE which are not blocked by buildings or the user's own body. For a dense urban area, we also analyze the impact of non-LOS (NLOS) links on blockage events. Our results indicate that the minimum density of BS required to satisfy the Quality of Service (QoS) requirements of Ultra Reliable Low Latency Communication (URLLC) applications will be driven mainly by blockage and latency constraints, rather than coverage or capacity requirements.

% Millimeter wave (mmWave) communication systems can provide high data rates, but the system performance may degrade significantly due to mobile blockers and the user’s own body. A high frequency of interruptions and long duration of blockage may degrade the quality of experience. Macrodiversity of base stations (BSs) has been considered a promising solution where the user equipment (UE) can handover to other available BSs if the current serving BS gets blocked. However, an analytical model to evaluate the system performance of dynamic blockage events in this setting is largely unknown. In this paper, we consider an open park-like scenario and obtain closed-form expressions for the blockage probability, the expected frequency, and duration of blockage events using stochastic geometry. Our results indicate that the minimum density of BS that is required to satisfy the Quality of Service (QoS) requirements of AR/VR and other Ultra Reliable Low Latency Communication (URLLC) applications is largely driven by blockage events rather than capacity requirements. As an alternative to increasing BS density, placing the BS at a greater height reduces the likelihood of blockage. We present a closed-form expression for the BS density height trade-off that can be used for network planning.
\end{abstract}

\IEEEpeerreviewmaketitle

\section{Introduction}
\label{sec:intro}

Recent advances in applications ranging from Machine-Type Communication (MTC) to mission-critical services (connected-vehicle-to-everything (V2X), eHealth, Augmented/Virtual Reality (AR/VR), and tactile Internet) are posing tremendous challenges regarding capacity, reliability, latency, and scalability. V2X, eHealth, and MTC may require low capacity, but impose a very tight constraint on the network for ultra-high service reliability (99.999\%), ultra-low latency ($1 - 10$ ms), and low Block Error Rates (BLER) ($10^{-9}-10^{-5}$)~\cite{mohr20165g, DBLP, ATTARVR}. On the other hand, AR/VR applications require a significantly higher capacity (100 Mbps - few Gbps) and low latency (1ms-10ms), while being able to tolerate a relatively higher BLER~\cite{DBLP}. Furthermore, most of these application will require a Service Interruption Time (SIT) of close to 0 ms~\cite{doetschdeliverable}. The stringent requirements of these applications propel the rethinking of 5G cellular network design for Ultra Reliable Low Latency Communication (URLLC) applications. While many proposals to achieve low latency and high reliability have been proposed, such as  edge caching, edge computing, network slicing~\cite{7387263,6871674}, a shorter Transmission Time Interval (TTI), frame structure~\cite{7980747}, and flow queueing with dynamic sizing of the Radio Link Control (RLC) buffer at Data Link Layer~\cite{Kumar2018DynamicCO}, we focus our analysis on issues related to Radio Access Network (RAN) planning to achieve the stringent Quality of Service (QoS) requirements of URLLC applications.

Millimeter wave (mmWave) frequencies are being considered for the 5G RAN due to their abundant bandwidth as compared to traditional sub-6 GHz bands~\cite{Rappaport2013Millimeter}. Thanks to the high bandwidth available at mmWave frequencies, the mmWave RAN can achieve data rates of the order of a few Gbps~\cite{CellularCap-Rap}, suitable for the QoS requirements for AR/VR applications. However, mmWave communication systems are quite vulnerable to blockages due to higher penetration losses and reduced diffraction~\cite{bai2015coverage}. Even the human body can reduce the signal strength by 20 dB~\cite{georgeFading}. Thus, an unblocked Line of Sight (LOS) link is highly desirable for mmWave systems. When a LOS link is blocked, strong Non-Line of Sight (NLOS) links may also be helpful in achieving high reliability. However, NLOS links can also be blocked by mobile blockers. A mobile human blocker can block a link 
for approximately 500 ms~\cite{georgeFading}. The frequent blockages of mmWave links and high blockage durations can be detrimental to URLLC applications.

One potential solution to blockages in the mmWave cellular network is employing a combination of macrodiversity of Base Stations (BSs) and Coordinated Multipoint (CoMP) techniques. These techniques have shown a significant reduction in interference and improvement in reliability, coverage, and capacity in the current sub-6 GHz Long Term Evolution-Advance (LTE-A) deployments~\cite{kim2011analysis}. Furthermore, RANs are moving towards the cloud-RAN (C-RAN) architecture that implements macrodiversity and CoMP techniques by pooling a large number of BSs in a single centralized Base Band Unit (BBU)~\cite{IBMCloudRAN,chen2011c}. As a single centralized BBU handles multiple BSs, the handover and beam-steering time can be reduced significantly~\cite{cloudRAN}. To reduce the SIT to close to 0 ms, the 3$^\text{rd}$ Generation Partnership Project (3GPP) has introduced Make-before-break (MBB) and Random Access Channel (RACH)-less techniques in Release 14~\cite{lte2017evolved}. In the MBB handover procedure, the connection to a serving BS is released only after the handover to the new BS is complete. In the RACH-less technique, the target BS obtains the necessary RACH information from the serving BS. Furthermore, to achieve a close to 0 ms handover latency, a multi-connectivity functional architecture is proposed, where different Radio Access Technologies (RATs) such as 5G New Radio (NR), LTE, and WiFi can be tightly coupled and a single handover command can be sufficient for service migration between different RATs and BSs~\cite{Multi2016Ravanshid}.  

%Furthermore, LTE BSs can help in achieving high reliability by serving as a backup in the case of a blockage of mmWave BSs. However, the LTE and 5G New Radio (NR) stacks must be coupled tightly for low handover delays. Note that the migration of services from 5G NR to LTE in case of blockages may cause up to 30 ms extra delay in the control plane~\cite{samsung2017}. The delay caused due to service migration from 5G NR to LTE can be reduced ($\sim 10$ ms) by considering shorter TTI and processing delay at BSs. Thus, even though LTE and 5G NR interworking may increase the service reliability, it fails to satisfy latency requirements of URLLC applications if the two stacks are not tightly-coupled.       

%Considering different URLLC applications and their QoS requirements, 
The impact of service interruptions due to blockage can also be alleviated by caching the downlink content at the BSs or the network edge for AR/VR applications~\cite{bastug2017toward}. However, caching the content more than about 10 ms may degrade the user experience and may cause nausea to the users particularly for AR applications~\cite{westphal2017challenges}. For applications like MTC, autonomous driving, eHealth, and the tactile Internet, the 5G network must achieve high reliability and low latency. Careful network planning can meet the requirements of URLLC applications.
%An alternative when blocked is to offload traffic to sub-6GHz networks such as 4G, but this needs to be carefully engineered to not overload them.
Therefore, it is important to study the blockage probability and blockage duration to obtain the optimal density of BS that satisfies the desired QoS requirements.

As shown in previous work by Bai, Vaze, and Heath~\cite{bai2012using,bai2014analysis2}, mmWave networks may not provide very high coverage due to static blockages (blockages such as buildings, trees, and other static structures). Thus, we envision the 5G cellular network as the overlap of sub-6 GHz LTE and 5G NR where the coverage holes of the 5G NR would be taken care of by LTE. In this architecture, both 5G and LTE will be collectively responsible for achieving the QoS requirements of URLLC applications and the C-RAN will be responsible for seamless migration of services from 5G NR to LTE (or vice-versa).
Our primary concern is to quantify the QoS requirements of URLLC applications in the mmWave cellular network where blockages may cause a significant problem.
To provide seamless connectivity for URLLC applications, we present an analysis considering key QoS parameters such as the \textit{probability of blockage events, blockage frequency, and blockage duration}.  
The major contributions of this paper are summarized as follows:

%For instance, to satisfy the QoS requirement of mission-critical applications such as AR/VR, tactile Internet, and eHealth applications, 5G cellular networks target a service reliability of 99.999\% \cite{mohr20165g}. 
%\textcolor{red}{ (insert some requirements for AR/VR requirements and give a citation.)}
% and the optimal amount of caching needed.    

% In this paper, we consider an open park scenario where every available BS in the vicinity of the UE are potential serving BSs to the UE.
% We consider a scenario where all the blockers in the system are mobile blockers. Although there are many

%This work presents a simple blockage model for the LOS link using tools from stochastic geometry. In particular, our contributions are as follow:

\begin{enumerate}
\item We provide a stochastic geometry based analytical model of the combined effect of static blockages (User Equipment (UE) blocked by permanent structures such as buildings), dynamic blockage (UE blocked by mobile blockers) and self-blockage (UE blocked by the user's own body) to evaluate the impact on key QoS metrics.
\item We obtain analytical expressions for the probability, frequency, and duration of simultaneous blockage of all BSs in the range of the UE. 
% We also compare two models, one based on stochastic geometry and the other based on a hexagonal cell layout. 
% For urban scenarios, we also consider the impact of static blockages and NLOS paths to obtain analytical expressions for blockage probability and duration.

% To measure the significance of NLOS path on the performance, we used random reflections from the buildings and other objected and modeled the number of reflected paths in the NLOS region around the user.   
\item We verify our dynamic blockage model through Monte-Carlo simulations by considering a random waypoint mobility model for mobile blockers. 
% \item For the case of hexagonal cellular cells, we showed that a well planned cellular architecture could combat blockages in mmWave cellular networks more efficiently. 
\item Finally, we present a case study to find the minimum BS density that is required to satisfy the QoS requirements of URLLC applications for two scenarios: an open park-like scenario and an urban scenario.
We also  analyze the trade-off between BS height and density to satisfy the QoS requirements.
\end{enumerate}

The rest of the paper is organized as follows. The related work is presented in Section \ref{sec:related-work}. The system model is described in Section \ref{sec:system-model}. Section~\ref{sec:analysis} provides an analysis of blockage events and evaluates the key blockage metrics for LOS link. The analysis is extended to consider blockage of both LOS and NLOS links in Section \ref{sec:NLOS}. Section \ref{sec:hex} considers the special case of a hexagonal cell layout. Numerical results are presented in Section~\ref{sec:numResults}.  
%The validation of our theoretical results with simulations and various case studies are presented in Section~\ref{sec:numResults}. 
% evaluates our theoretical results with MATLAB simulations using random waypoint mobility model.
Finally, Section \ref{sec:conclusion} concludes the paper.

\section{Related Work}
\label{sec:related-work}
% The mmWave blockage has been extensively studied for the coverage and capacity analysis. 

% Another approach is to \textbf{collect real-world data} on received signal and fit a blockage model. However, the model is very specific to the considered scenario and may not be generalized well.

% The third approach is to study the impact of blockages through \textbf{analytical approach} such as stochastic geometry ~\cite{wang2017blockage} and Markov models.

A mmWave link may have three kinds of blockages, namely, static, dynamic, and self-blockage. Static blockage due to buildings and permanent structures has been thoroughly studied by Bai, Vaze, and Heath in~\cite{bai2012using} and~\cite{bai2014analysis2} using random shape theory and a stochastic geometry approach for urban microwave systems. The underlying static blockage model is incorporated into the cellular system coverage and rate analysis by Bai \textit{et al.}~\cite{bai2015coverage}. Static blockages from permanent structures may cause a significant reduction in LOS link quality. On the other hand, reflections from such structures can provide sufficient NLOS paths to recover from the LOS path loss. Bai \textit{et al.}~\cite{bai2015coverage} modeled the LOS and NLOS BSs as independent Poisson point processes and showed rate and coverage gain in the mmWave cellular network. A detailed NLOS model is presented in \cite{dong2016reliability} by considering first-order reflections. However, they consider the blockages as squares with a fixed orientation. Akdeniz \textit{et al.}~\cite{CellularCap-Rap} obtained a statistical model of the number of NLOS BSs through measurements in an urban scenario. Note that for an open area such as a public park, static blockages and NLOS paths play a small role.     

The second type of blockage is dynamic blockage due to mobile humans and vehicles (collectively called mobile blockers) which may cause frequent interruptions to the LOS link. Dynamic blockage has been given significant importance by 3GPP in TR 38.901 of Release 14~\cite{3gpptr}. An analytical model in~\cite{gapeyenko2016analysis} considers a single access point, a stationary user, and blockers located randomly in an area. The model in~\cite{wang2017blockage} is developed for a specific scenario of a road intersection using a Manhattan Poisson point process model.
MacCartney \textit{et al.}~\cite{georgeFading} developed a Markov model for blockage events based on measurements on a single BS-UE link. Similarly, Raghavan \textit{et al.}~\cite{raghavan2018statistical} fits the blockage measurements with various statistical models. However, a model based on experimental analysis is generally specific to the measurement scenario and may not generalize well to other scenarios. The authors in~\cite{han20173d} considered a 3D blockage model and analyzed the blocker arrival probability for a single BS-UE pair. 
% Other dynamic blockage models consider a detailed model for single BS-UE link~\cite{eliasi2015stochastic}, based on either measurement or ray-tracing data. 
Studies of spatial correlation and temporal variation in blockage events for a single BS-UE link are presented in~\cite{samuylov2016characterizing} and~\cite{gapeyenko2017temporal}. However, their analytical model is not easily scalable to multiple BSs.
% , important when considering the impact of macrodiversity.

Apart from static and dynamic blockage, self-blockage plays a significant role in mmWave systems performance.
% The \textbf{self-blockage }is also important and is studied in.
The authors of~\cite{abouelseoud2013effect} studied human body blockage through simulation. A statistical self-blockage model is developed in~\cite{raghavan2018statistical} through experiments considering various modes (landscape or portrait) of hand-held devices. The impact of self-blockage on received signal strength is studied by Bai and Heath in~\cite{bai2014analysis} through a stochastic geometry model. They assume the self-blockage due to a user's body blocks the BSs in an area represented by a cone.

All the above blockage models consider the UE's association with a single BS. 
% Blockage models for a single BS-UE link has been developed in . 
% These models try to capture the specific details of the blockage event, but don't provide insight into the blockage events of multiple BSs.
%In this paper, we consider %\textbf{macrodiversity of BSs}, 
Macrodiversity of BSs is considered as a potential solution to alleviate the effect of blockage events in a cellular network.
%such that the UE can connect to any unblocked BS in case the currently serving BS gets blocked.
The authors of~\cite{zhu2009leveraging} and~\cite{zhang2012improving} proposed an architecture for macrodiversity with multiple BSs and showed improvements in network throughput. A blockage model with macrodiversity is developed in~\cite{choi2014macro} for independent blocking and in~\cite{gupta2018macrodiversity} and \cite{gupta2018impact} for correlated blocking. However, they consider only static blockage due to buildings. 
%and do not consider mobile blockers. 

%A detailed methodology for BS deployment is proposed in~\cite{petrov2017dynamic} based on a precise 3D map of the area and a ray-tracing model.  However, an accurate simulation requires the tracking of all the moving objects in an area. This process is computationally intensive and requires specific knowledge of the terrain information. On the other hand, 
The primary purpose of the blockage models in previous papers was to study the coverage and capacity analysis of the mmWave system. 
However, apart from signal degradation, blockage frequency and duration also affect the performance of the mmWave system and are critically important for URLLC applications. 

In our previous work~\cite{jain2018driven}, we presented the effect of dynamic blockage and self-blockage on the LOS link in an open park scenario. The results presented in~\cite{jain2018driven} indicate that a high density of BS is required to satisfy the QoS requirements of URLLC applications. The study motivated us to consider a regular hexagonal cell layout in this paper. We showed that a well-planned cellular network such as the hexagonal layout could combat blockage events more efficiently than a random BS deployment in an open park scenario. Furthermore, in this paper, we also consider an urban setting, where, apart from dynamic and self-blockage, there are static blockages which may block the LOS paths between UE and BSs, while at the same time provide NLOS paths between UE and BSs. %\textcolor{red}{Considering the fact that in both urban environment and indoor environment, there may exist static blockage and NLOS path, we consider applications like  connected-vehicle-to-everything communication (V2X) in cities, AR/VR in dense urban environment, eHealth, indoor applications such as M2M communication, indoor AR/VR, tele-surgery and industrial automation under the umbrella of urban scenario.} 
This paper extends our previous dynamic blockage analysis to include static blockages and derives the blockage probability and duration considering both LOS and NLOS paths.

\begin{table}[!t]
%\vspace{2mm}
\caption{Summary of Notations}
\label{tab:notations}
\centering
\begin{tabular}{|p{3.5em}|p{24.5em}|}\hline
	\textbf{Notation} & \textbf{Description} \\\hline 
     $R$ & LOS Range. \\
     $\Rt$ &  NLOS Range. \\
     $\lambda_T$&BS Density.\\
     $\lambda_B$, $\lambda_S$ & Density of dynamic blockers, and static blockages resp. \\
    $\alpha_i$ & Arrival rate of dynamic blockers.\\
    $\omega$ & Self-blockage angle. \\ 
    $m$& Number of BS in disc $B(o,R)$.\\
    $n$& Number of BS in coverage.\\
    $k$& Number of NLOS links for a given BS.\\
    $\kappa$ & Parameter for the distribution of number of NLOS paths. \\
    $r_i$&UE-BS distance for $i$th BS. \\
    $\setC^{d}$& LOS coverage considering only self-blockage for open park scenario.\\
    $\setC^{LOS}$& LOS Coverage considering static and self-blockage.\\
    $\setC$& LOS and NLOS coverage considering static and self-blockage.\\
    $B^s,B^d$& Indicator for static and dynamic blockage respectively for a single BS-UE link.\\
    $B^{\text{self}}$& Indicator for self-blockage.\\
    $B^{\text{hex}}$& Blockage indicator for hexagonal case. \\
    $B^{LOS}$& LOS Blockage indicator considering static, dynamic, and self-blockage.\\
    $B$& LOS and NLOS blockage indicator considering static, dynamic, and self-blockage.\\
    $\zeta^d$& Blockage frequency considering only dynamic and self-blockage.\\
    $T^{d}$& Blockage duration of LOS paths considering dynamic and self-blockage.\\
    $T^{LOS}$& Blockage duration of LOS paths considering static, dynamic, and self-blockage without NLOS paths.\\
    $T$& Blockage duration considering static, dynamic, and self-blockage with both LOS and NLOS paths.\\
\hline
\end{tabular}
\vspace{-7mm}
\end{table}

%!!!!!!!!!!!!!!!!-----------------------!!!!!!!!!!!!!!!!!!
\section{System Model}
% We discuss the various component of our system model in this section.

\begin{figure}[!t]
	\centering
	\includegraphics[width=0.45\textwidth]{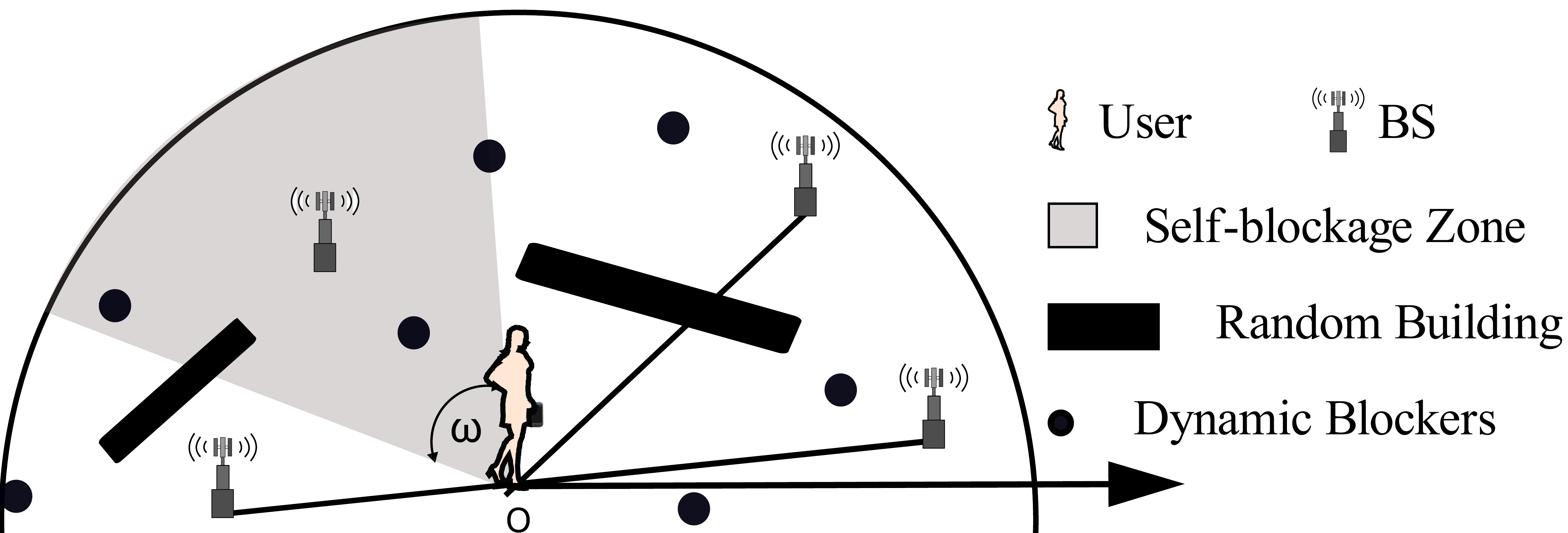}
    \caption{System Model: Service quality to the UE can be degraded by static blockage, self-blockage, and dynamic blockages. A typical UE is located at the center and BSs and blockers are located uniformly in a disc around the UE.}
\label{fig:sysMod}
\vspace{-4mm}
\end{figure}
\label{sec:system-model}
Various components of our system model and the associated assumptions are presented in this section. The system model is shown in Figure \ref{fig:sysMod}.
\subsection{Assumptions} 
\begin{itemize}
\item \textit{BS Model}: The mmWave BS locations are modeled as a homogeneous Poisson Point Process (PPP) with density $\lambda_T$ (a brief summery of notations is provided in Table \ref{tab:notations}).
% A typical UE (UE) is located at origin  $o$. 
% Let $\Lambda=\{\mathbf{x}_i,i\in \mathbb{N}\}$ denotes the BS PPP with location $\mathbf{x}_i$ and distance $r_i$ from the origin $o$.
Consider a disc $B(o,R)$ of radius $R$ and centered around the origin $o$, where a typical UE is located.
We assume that each BS in $B(o,R)$ is a potential serving BS for the UE. 
%\textcolor{blue}{However, at a time the UE can be served by either a single BS or multiple BSs depending upon hardware capability of the UE's device. Further, assume that in case of UE can be served by only one BS, beam-steering and handover time as approximately zero.}
Thus, the number of BSs $M$ in the disc $B(o,R)$ of area $\pi R^2$ follows a Poisson distribution with parameter $\lambda_T\pi R^2$, i.e.,
% Now, consider an event $\mathcal{N}$ that $N$ BSs are inside this disc of area $\pi R^2$, then the PDF of the corresponding random variable $N$ is:
\begin{equation}\label{eqn:poisson}
    P_M(m) = \frac{[\lambda_{T} \pi R^2]^m}{m!}e^{-\lambda_{T} \pi R^2}.
\end{equation}  
% \textcolor{red}{where the symbol `m!' means the factorial of m.}
% Let $P_O$ indicates the probability of UE is in out of coverage of all of BSs, i.e., there is no BS in the disc $B(o,R)$. Then, from~(\ref{eqn:poisson}), we can write
% \begin{equation}\label{eqn:outofCov}
% P_O = P_N(n=0) = e^{-\lambda_{T} \pi R^2}.
% \end{equation}

Given the number of BSs $m$ in the disc $B(o,R)$, we have a uniform probability distribution for BS locations.
The BSs distances $R_i$, $\forall i=1,\ldots,m$, from the UE are independent and identically distributed (i.i.d.) with distribution
% Therefore, the distance distributions are:
\begin{equation}\label{eqn:distribution}
 f_{R_i|M}(r|m) = \frac{2r}{R^2}; \ 0< r\le R, \forall i=1,\ldots,m,
%     f_{r_n|N}(r_n|n) = \frac{2r_n}{R^2}\quad0\le r\le R,
\end{equation}
%\textcolor{blue}{[Reviewer3]: BS model: distribution of distances of BSs are clear, but how are their angles of their position distributed (I assume uniformly over the arc)?}
and the angular positions $\theta_i, \forall \;i=1,\ldots,m$ of the BSs with respect to the x-axis are i.i.d. and follow a uniform distribution in $[0,2\pi]$.

%\textcolor{red}{and the angle $\theta_i$ that the BS makes from x-axis is distributed uniformly in $[0,2\pi]$.}
% distance of BS from the UE and $\mathbb{N}$ is the set of BSs those can satisfy QoS requirements of the UE.

%\begin{figure}[!t]
%	\centering
%	\includegraphics[width=0.3\textwidth]{img/self-block}
%	\caption{Self-blockage: Due to the orientation of UE and the user's body, a fraction of mmWave BSs will be blocked. The self-blockage-zone is represented by a sector $\mathcal{S}(o,R,\omega)$ of angle $\omega$.}
%	\label{fig:selfblock}
%	\vspace{-6mm}
%\end{figure}

\item \textit{Static Blockage Model}: Static blockage due to buildings and permanent structures is thoroughly studied by Bai and Heath~\cite{bai2014analysis2}. They modeled static blockages such as buildings using random shape theory. The buildings are considered to be of length $\ell$, and width $w$. The probability\footnote{Note that the notation $P(B^s_i|m,r_i)$ is a compressed version of $P_{B^s_i|M,R_i}(B^s_i|m,r_i)$, which we use for convenience. Similar notations used elsewhere in this paper should be clear based on the context.}
  $P(B^s_i|m,r_i)$ that the $i$th BS-UE link is blocked by a static blockage is calculated in~\cite{bai2014analysis2} as
\begin{equation}\label{eqn:pBstatic}
P(B^s_i|m,r_i) = 1-e^{-(\beta r_i+\beta_0) };\quad \forall i=1,\cdots m,
\end{equation}
where $\beta = \frac{2}{\pi}\lambda_S(\E[\ell]+\E[w])$ and $\beta_0=\lambda_S\E[\ell]\E[w]$; where $\lambda_S$ is the density of static blockage represented in terms of static blockages per km$^2$ (sbl/km$^2$), and $\E[\ell]$ and $\E[w]$ are the expected length and width of buildings respectively. The effect of the height of buildings on static blockage probability is analyzed in~\cite{bai2014analysis2}. For simplicity, we assume the static blockages are higher than the BS. 

\item \textit{Self-blockage Model}: The user blocks a fraction of BSs due to his/her own body. The self-blockage zone is defined as a sector of the disc $B(o,R)$ making an angle $\omega$ at the user's body as shown in Figure~\ref{fig:sysMod}. The orientation of the user's body is uniform in [$0,2\pi$]. We consider a BS is blocked by self-blockage if it lies in the self-blockage zone. %\textcolor{blue}{[R3]: eq (4): consistently, you should give a definition like \(P(B...) = p = 1-\omega/2\pi\)}
Therefore, the probability that a randomly chosen BS is blocked by self-blockage is
\begin{equation}\label{eqn:self1}
P(B^\text{self})=\frac{\omega}{2\pi}.
\end{equation}
For ease of notation, we denote by $p$ the probability that a randomly chosen BS is not blocked by self-blockage:
\begin{equation}\label{eqn:self}
p=1-P(B^\text{self})=1-\frac{\omega}{2\pi}.
\end{equation}
% \textcolor{red}{We assume that the self-blockage zone is not changing with time, i.e., the orientation of user's body is assumed to be a constant.}

\item \textit{Dynamic Blockage Model}: 
Dynamic blockers are distributed according to a homogeneous PPP with density $\lambda_{B}$ represented in terms of blockers per m$^2$ (bl/m$^2$). The blockers\footnote{In the rest of the paper, the term \textit{blocker} implies dynamic blockers, unless otherwise stated.} are moving in the disc $B(o,R)$ with velocity $V$ in a random direction. Further, the arrival process of the blockers crossing the $i$th BS-UE link is Poisson with intensity $\alpha_i$. The blockage duration is independent of the blocker arrival process and is exponentially distributed with mean $1/\mu$. %\textcolor{blue}{[R3]: dyn. blockage model: please state the formula of the distribution and clarify all used parameters}
Thus, the blocking event follows an exponential on-off process with $\alpha_i$ and $\mu$ being the blocking and unblocking rates, respectively. The probability of blockage $P(B^d_{i}|m,r_i)$ of the $i$th BS-UE link due to dynamic blockers can therefore be expressed as follows:
\begin{equation}\label{eqn:AllBSdynamic1}
P(B^d_{i}|m,r_i)=\frac{1/\mu}{1/\alpha_i + 1/\mu}=\frac{\alpha_i}{\alpha_i + \mu}.
\end{equation} 

We will derive an expression for the blocker arrival rate $\alpha_i$ in Lemma \ref{lemma:alphan} in the next subsection. 

\begin{figure}[!t]
    \centering
    \includegraphics[width=0.405\textwidth]{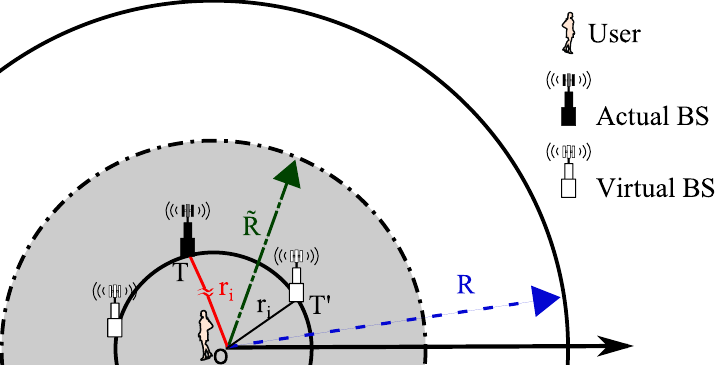}
    \caption{NLOS model: An actual BS located in a disc $B(o,\Rt)$ of radius $\Rt$ (shown as shaded region) may provide strong NLOS signals. For a given BS at a distance $r_i$ from the origin, the corresponding virtual BSs are located on a circle of radius $r_i$ around the origin.}
    \label{fig:NLOS_analysis_approx}
    \vspace{-4mm}
\end{figure}

% A blocker can change the direction of motion at any time. However, we assume that for a small time interval $\Delta t$, the direction of blockers' motion remains consistent. As the future positions of blockers only depend upon their most recent positions, we modeled blocker's arrival process as a Markov process with the blockage rate $\alpha$. 

% \item \textit{Blockage duration Model}: We consider a simple scenario where the blockers are modeled as cylinders of constant height, but with a varying diameter. \textcolor{red}{(Need to rewrite)}
% The cumulative effect of multiple blockers can be approximated by an exponentially distributed blocking period with mean $1/\mu$. Where $\mu$ is the mean service rate of each blockage event. 

%We compared the results of this simple exponential on-off model with the $M/GI/\infty$ model~\cite{gapeyenko2017temporal} in Figure \ref{fig:fracB} and found that for a blocking rate $\mu=1.9$, the two results are similar. This justifies the validity of our simple model.

\item \textit{NLOS Model}: NLOS links can fill the coverage holes in the mmWave cellular system and can enhance reliability. However, the signal degradation, due to reflections from reflectors such as buildings, may limit the number of strong NLOS paths. Since the NLOS link length is always longer than the corresponding LOS link length, we only consider those NLOS paths which correspond to BSs within $\Rt \le R$ distance from the origin. 
The corresponding disc $B(o,\Rt)$ is shown in Figure \ref{fig:NLOS_analysis_approx}. 
The number of NLOS links $K$ for a given BS-UE pair was obtained by Akdeniz \textit{et al.}~\cite{CellularCap-Rap} through measurements in an urban area as
\begin{equation}
K\sim\max\{\text{Poisson}(\kappa),1\},
\end{equation}
where $\kappa$ is obtained empirically. %\textcolor{red}{Note that the symbol $\sim$ denotes distributed according to a given probability distribution. } 
The distribution $P_K(k)$ follows:
\begin{equation}\label{eqn:K}
P_K(k) =
\begin{cases}
0 & \textrm{if } k=0,\\
e^{-\kappa}+e^{-\kappa}\kappa & \textrm{if } k=1,\\
\frac{e^{-\kappa}\kappa^k}{k!} & \textrm{if } k>1.
\end{cases}
\end{equation}
We consider virtual locations of NLOS BSs as images of LOS BS via reflectors. For a given BS at a distance $r_i$ from origin, the corresponding virtual BSs would be at a distance longer than $r_i$. To obtain a lower bound on blockage probability, we consider the best case scenario when the NLOS BSs are located exactly at a distance $r_i$ from origin. For simplicity, we assume the virtual BSs are distributed uniformly in a ring of radius $r_i$, which is the same as the actual BS-UE link length.

% we only consider the locations of virtual BS within a range of $\Rt$ from the UE. For the $i$th BS, we consider the location of its $K$ virtual BSs is uniform  in an annular region as shown in Figure \ref{fig:NLOS_analysis_approx}. Thus the distribution of distance $\rt_{i,j}$ from the UE to the $j$th virtual BS corresponding to the $i$th LOS BS is i.i.d. for $j=1,\cdots,k$, and is given by
% \begin{equation}\label{eqn:rtij}
% \begin{split}
% f(\rt_{i,j}|m,r_i,k) &= \frac{2\rt_{i}}{\Rt^2-r_i^2};\quad \textrm{for}\; r_i<\rt_i<\Rt\\
% &\forall i=1,\cdots, m; \quad \forall j=1,\cdots,k.
% \end{split}
% \end{equation}

\item \textit{Connectivity Model}: We say the UE is blocked when all of the potential serving BSs, actual and virtual, in the disc $B(o, R)$ are blocked simultaneously.

% if all the BSs those can satisfy QoS requirement of the UE are blocked, simultaneously.
\end{itemize}

%\begin{figure}[!t]
%	\centering
%	\includegraphics[width=0.45\textwidth]{img/blockerModel}
%	\caption{System Model}
%	\label{fig:sysMod}
%	\vspace{-4mm}
%\end{figure}

 %\begin{figure}
  %   \centering
   %  \includegraphics[width=0.40\textwidth]{img/FracBlockage}
    % \caption{Fraction of blockage for $M/GI/\infty$ process~\cite{gapeyenko2017temporal} compared with exponential on-off process with $\mu=1.9\;$ and varying $\lambda_B$.}
     %\label{fig:fracB}
     %\vspace{-6mm}
 %\end{figure}

\begin{figure}[!t]
	\centering
	\includegraphics[width=0.5\textwidth]{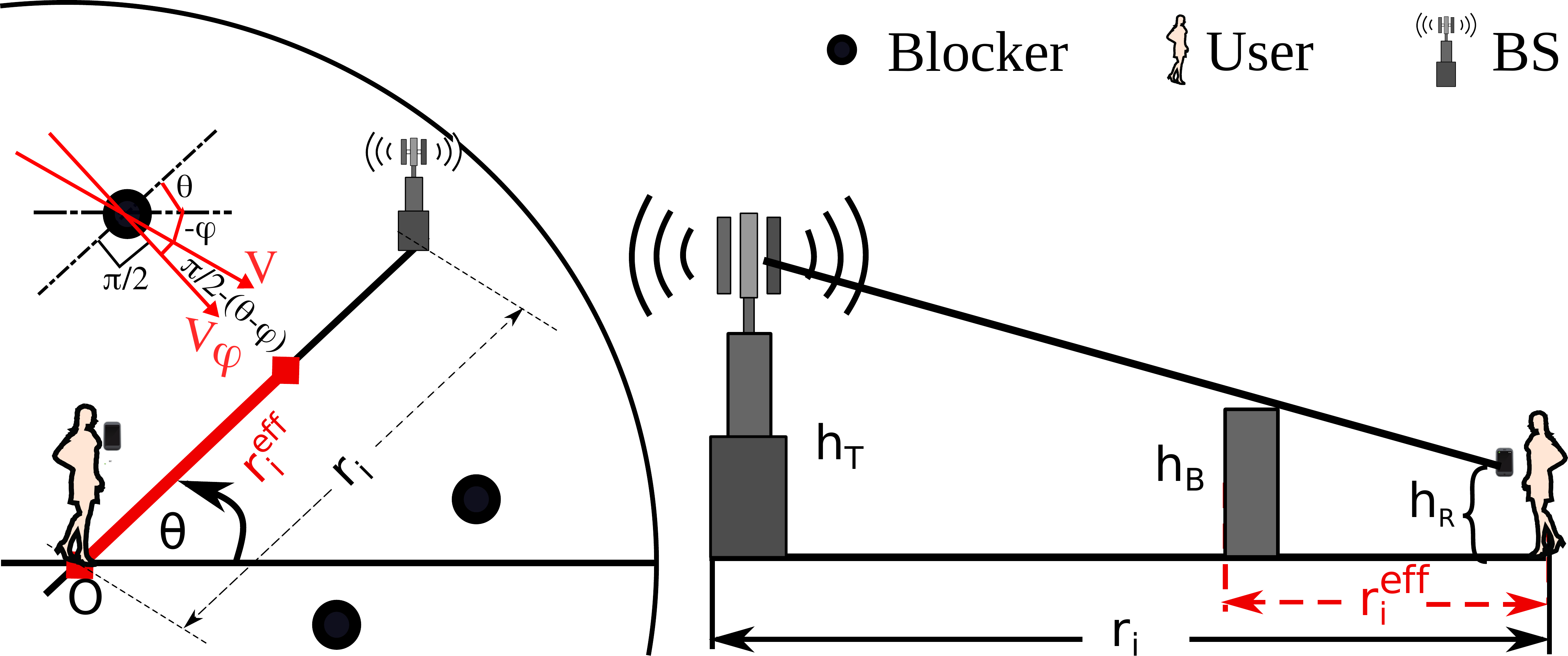}
    \caption{Dynamic blockage model: Blockers are moving with velocity $V$ at a random angle $\varphi$ from the x-axis. A blocker within $r^{\text{eff}}_i$ (a fraction of the BS-UE distance $r_i$) distance away from the UE can cause interruptions to the $i$th BS-UE link.}
    \label{fig:dynblck}
\end{figure}
\subsection{Dynamic Blockage Model for a Single BS-UE Link}
For a sound understanding of the dynamic blockage model, consider a single BS-UE LOS link in Figure~\ref{fig:dynblck}. The 2D distance between the $i$th BS and the UE is $r_i$, and the LOS link makes an angle $\theta$ from the positive x-axis in the azimuth plane. Further, the blockers in the region move with constant velocity $V$ at an angle $\varphi$ with the positive x-axis, where $\varphi$ is distributed uniformly in $[0,2\pi]$. 
Note that only a fraction of blockers crossing the BS-UE link will be blocking the LOS path, as shown in Figure~\ref{fig:dynblck}. The effective link length $r_i^{\text{eff}}$ that is affected by the movement of dynamic blockers is given by,
% We define the blockage zone as a fraction of the BS-UE link such that only a blocker in the blockage zone will block the LOS signal. The length of the blockage zone is evaluated as
\begin{equation}\label{eqn:rieff}
    r_i^{\text{eff}}=\frac{\left(h_B-h_R\right)}{\left(h_T-h_R\right)}r_i,
\end{equation}
where $h_B$, $h_R$, and $h_T$ are the heights of blocker, UE (receiver), and BS (transmitter) respectively. The blocker arrival rate $\alpha_i$ (also called the blockage rate) is evaluated in Lemma~\ref{lemma:alphan}. 
% is called the blockage rate and is evaluated in Lamma 1. 
% The blockage rate is evaluated in Lemma 1.

\begin{lemma}\label{lemma:alphan}
 The blockage rate $\alpha_i$ of the $i$th BS-UE link is
% The rate at which the blockers cross the BS-UE link is defined as blockage rate and is given by
\begin{equation}\label{eqn:SingleBS}
\alpha_{i}=Cr_i,\quad\forall i=1,\cdots,m,
\end{equation}
where $C$ is proportional to blocker density $\lambda_B$ as,
\begin{equation}\label{eqn:C}
C = \frac{2}{\pi}\lambda_{B} V\frac{(h_B-h_R)}{(h_T-h_R)}.
\end{equation} 

\end{lemma}
\begin{proof}
Consider a blocker moving at an angle $(\theta-\varphi)$ relative to the BS-UE link (See Figure~\ref{fig:dynblck}). The component of the blocker’s velocity perpendicular to the BS-UE link is $V_\varphi = V\sin(\theta-\varphi)$, where $V_\varphi$ is positive when $(\theta-\pi)<\varphi<\theta$. 
% angle of motion of the blocker perpendicular to the line of connectivity of UE and BS is given by $(\theta-\varphi)$. Thus, the component of blocker’s velocity perpendicular to the BS-UE link is $V_\varphi = V\sin(\theta-\varphi)$, where $V_\varphi$ is positive when $(\theta-\pi)<\varphi<\theta$. 
% i.e., as relative height of BSs can be quite high compared to blocker, only a fraction of blockers crossing the link close to the UE would be able to block LOS link. Equivalently, we can say that if a blocker cross the LOS path at a distance $r_n^{eff}=\frac{r_n\left(h_B-h_R\right)}{\left(h_T-h_R\right)}$, it will block the UE~\cite{gapeyenko2017temporal}. Here, $h_B$ is the height of blocker, $h_R$ is the height of UE or receiver, and $h_T$ is the height of transmitter or BSs. 
Next, we consider a rectangle of length $r_i^{\text{eff}}$ and width $V_\varphi\Delta t$. The blockers in this area will block the LOS link over the interval of time $\Delta t$. Note there is an equivalent area on the other side of the link. Therefore, the average frequency of blockage is $2\lambda_{B}r_i^{\text{eff}}V_\varphi\Delta t = 2\lambda_{B} r_i^{\text{eff}}V\sin(\theta-\varphi)\Delta t$. Thus, the frequency of blockage per unit time is $2\lambda_{B} r_i^{\text{eff}}V\sin(\theta-\varphi)$. Taking an average over the uniform distribution of $\varphi$ (uniform over $[0,2\pi]$), we get the blockage rate $\alpha_i$ as follows:
\begin{equation}\label{eqn:SingleBS1}
\begin{split}
\MoveEqLeft\alpha_{i} = 2\lambda_{B} r_i^{\text{eff}}V\int_{\varphi = \theta-\pi}^{\theta}\sin(\theta-\varphi)\frac{1}{2\pi}\,d\phi \\
\MoveEqLeft \qquad \quad =\frac{2}{\pi}\lambda_{B} r_i^{\text{eff}}V = \frac{2}{\pi}\lambda_{B} V\frac{(h_B-h_R)}{(h_T-h_R)}r_i.
\end{split}
\end{equation}
This concludes the proof. 
\end{proof}
% \textcolor{red}{Note that each blockage event will initiate the handover processes, thus the UE handover rate  }

Following~\cite{gapeyenko2017temporal}, we model the blocker arrival process as Poisson with parameter $\alpha_i$ blockers/sec (bl/s).
Note that there can be more than one blocker simultaneously blocking the LOS link. The overall blocking process has been modeled in~\cite{gapeyenko2017temporal} as an alternating renewal process with alternate periods of blocked/unblocked intervals, where the distribution of the blocked interval is obtained as the busy period distribution of a general $M/GI/\infty$ queuing system. Here $M$ stands for Poisson arrival process of blockers and $GI$ denotes a general independent distribution of blockage durations (service time), which depends on the velocity and direction of the arrival of the blockers. Since there can be many independent  blockage events overlapping in time, we assume infinite servers.
% Number of servers is assumed to be infinite, i.e., when a blocker enters the blockage zone, it immediately starts to block the link (start getting served, i.e., the blockage duration due to a particular blocker mobility is not dependent upon another).

For mathematical simplicity, we assume the blockage duration of a single blocker is exponentially distributed with parameter $\mu$. Thus we can model the blocker arrival process as an $M/M/\infty$ queuing system. We further approximate the overall blocking process as an alternating renewal process with exponentially distributed periods of blocked and unblocked intervals with parameters $\alpha_i$ and $\mu$ respectively. 
Equivalently, the simultaneous blocking by two or more blockers is assumed to be a negligible probability event. 
This approximation works for a wide range of blocker densities, as is verified through simulations in Section \ref{sec:numResults}.
% This approximation is justified in Appendix \ref{app:mmoo_approx}             

% In general, any number of blockers may simultaneously block the BS-UE LOS link and the overall blockage period is evaluated as the busy period of $M/M/\infty$ system~\cite{gapeyenko2017temporal}.  

% However, for mathematical simplicity, we approximate the $M/M/\infty$ system with an alternating renewal process with exponentially distributed block and unblock interval. In the alternating renewal process each block and unblock event is followed by each other. At the same time, each block and unblock duration is exponential in length. For this to happen, we require low probability of simultaneous blockage of a given BS-UE link. 
The blocking event of a BS-UE link follows an on-off process with $\alpha_i$ and $\mu$ as blocking and unblocking rates, respectively. The probability of blockage $P(B^d_{i}|m,r_i)$ of the $i$th BS-UE link due to dynamic blockers follows:
\begin{equation}\label{eqn:AllBSdynamic}
P(B^d_{i}|m,r_i)=\frac{\alpha_i}{\alpha_i + \mu} = \frac{\frac{C}{\mu}r_i}{1+\frac{C}{\mu}r_i},\quad\forall i=1,\cdots,m,
\end{equation} 
where $C$ is proportional to the blocker density $\lambda_B$ (see (\ref{eqn:C})). 
%For readers interested in obtaining state probabilities other than blockage, a Markov chain analysis for dynamic blockages is presented in Appendix \ref{app:Markov_NBS}.

In the next section, we consider a generalized LOS blockage model considering all three kind of blockages. We assume the UE keeps track of all available BSs using beam-steering and handover protocols, which we assume can be achieved with a SIT of 0 ms, as discussed in Section I.
\section{Generalized LOS Blockage Model}\label{sec:analysis}
%\textcolor{blue}{[R3]: Please add some short overview of the section before A. Coverage Probability in order to provide some guidance.} 
% We present a generalized LOS blockage model considering the impact of static blockages such as buildings, dynamic blockers such as people and vehicles and self-blockage by user's own body. 
In this Section, we evaluate the impact of mobile blockers such as people and vehicles on the direct LOS BS-UE link. We consider macrodiversity of BS and evaluate the blockage probability given that at least one of the BSs-UE links is available. A link is said to be available if it is not permanently obstructed by static blockages such as buildings and the user's own body.
%In the open park scenario, static blockers such as buildings and permanent structures do not exists. Thus, the mmWave links do not suffer static blockages and NLOS paths do not exists at the same time. Thus, for open park scenario, we present only the blockages of LOS paths in this section. 
\subsection{LOS Coverage Probability}
We say that the UE is in coverage if there is at least one BS not blocked by either static blockage or by self-blockage.
% Let there be $N$ BSs within the range of the UE that are not blocked by static and self-blockage. 
 \begin{lemma}\label{lemma:PN}
The distribution of the number of BSs ($N$) which are not blocked by static or self-blockage follows a Poisson distribution with parameter $pq\lambda_{T} \pi R^2$, i.e.,
% outside the self-blockage zone and in the disc $B(o,R)$ is
\begin{equation}\label{eqn:PN}
P_N(n) = \frac{[pq\lambda_{T} \pi R^2]^n}{n!}e^{-pq\lambda_{T} \pi R^2},
\end{equation}
where, \begin{equation}\label{eqn:p}
p= 1-\omega/2\pi,\; \text{and}, \;q = \frac{2e^{-\beta_0}}{\beta^2R^2}\left[1-(1+\beta R)e^{-(\beta R)}\right].
\end{equation} 
% is the probability that a randomly chosen BS lies outside the self-blockage zone in the disc $B(o,R)$ (the area of such region would be $p\pi R^2$).
\end{lemma}
\begin{proof}
See Appendix \ref{app:proof_PNn}.
\end{proof}
\begin{corollary}
Let $\setC^{LOS}$ denote the LOS coverage event that represents that the UE has at least one serving BS in the disc $B(o,R)$ and is not blocked by static or self-blockage. The probability of this LOS coverage event $\setC^{LOS}$ is calculated as
\begin{equation}\label{eqn:coveragep}
\begin{split}
P(\setC^{LOS}) &=P_N(n\ne 0)=1-P_N(0)\\&=1- e^{-pq\lambda_T\pi R^2},
\end{split}
\end{equation}
where $p$ and $q$ are defined in (\ref{eqn:p}).
\end{corollary}

\vspace{-4mm}
\subsection{LOS Blockage Probability}
% Given that there are $n$ BSs in the communication range of the UE and are not blocked by any buildings or the user's body, they can still get blocked by mobile blockers. The blocking event of these $n$ BS-UE links is assumed to be independent on-off processes with $\alpha_i$ and $\mu$ as blocking and unblocking rates, respectively. 
% The probability of blockage of the $i$th BS-UE link is $\alpha_i/(\alpha_i + \mu)$. 
Our objective is to develop a blockage model for the mmWave cellular system where the UE can connect to any of the potential serving BSs. In this setting, a blockage occurs when all the BS-UE links are blocked simultaneously. 
% Since we assume independence, the probability that all are in blockage is their product.
We define an indicator random variable $B^{LOS}$ that indicates the LOS blockage of all available BSs in the range of the UE. 
% The blockage occurs when the BS-UE link is blocked by either static, dynamic or self-blockage. 
We obtain the blockage probability by utilizing the independence of static, dynamic, and self-blockage events. 
A BS is considered blocked if it is blocked by either static blockage or dynamic blockage or self-blockage. The LOS blockage probability $P(B_i^{LOS}|m,r_i)$ of the $i$th BS-UE link is conditioned on the number of BSs $m$ in the disc $B(o,R)$ and the BS-UE distance $r_i$, and is given by:
\begin{equation}\label{eqn:blProbLOSstep1}
\begin{split}
P(B_i^{LOS}|m,r_i) &= 1-\left(1-P(B^{\text{self}})\right)\left(1\!-\!P(B^{s}_i|m,r_i)\right)\\&\qquad\qquad\qquad\qquad\quad\times\left(1\!-\!P(B^{d}_i|m,r_i)\right)\\
&=1-pe^{-(\beta r_i+\beta_0)}\frac{1}{1+\frac{C}{\mu}r_i},
\end{split}
\end{equation}

where we use the expressions for static and dynamic blockage probability from (\ref{eqn:pBstatic}) and (\ref{eqn:AllBSdynamic}) respectively. The blockage events of multiple BSs may be correlated based on the location of dynamic blockers. For instance, a blocker closer to the user may simultaneously block multiple BSs as compared to a blocker farther away from the the user. However, in this paper, we assume independent blockage of all the BS-UE links and leave the correlation analysis for future work. Denoting the set of distances $\{r_1, r_2,\cdots,r_m\}$ by a compressed form $\{r_i\}$, we evaluate the probability of simultaneous blockage of all the LOS BS-UE links, i.e., $P(B^{LOS}|m,\{r_i\})$, as follows:
% \footnote{We use $\{r_i\}$ to denote the set $\{r_1,r_2,\cdots,r_m\}$. }
\begin{equation}\label{eqn:blProbLOSstep2}
\begin{split}
P(B^{LOS}|m,\{r_i\})
&= \prod_{i=1}^m P(B_i^{LOS}|m,r_i)\\
&=\prod_{i=1}^m\left(1-pe^{-(\beta r_i+\beta_0)}\frac{1}{1+\frac{C}{\mu}r_i}\right).
\end{split}
\end{equation}

% which are not blocked by buildings and the user's body and the link lengths $\{r_i\}$ for $i=1,\cdots,n$.
% Since we assume independent blockage of BS-UE links, $P(B_d|n,\{r_i\})$ can be written as the product of individual blockage probabilities as
% \begin{equation}\label{eqn:AllBS}
%         P(B_d|n,\{r_i\})= \prod_{i=1}^n\frac{\alpha_i/\mu}{1+\alpha_i/\mu}   =\prod_{i=1}^n\frac{\frac{C}{\mu}r_i}{1+\frac{C}{\mu}r_i},
% \end{equation} 
% where $C$ is defined in (\ref{eqn:C}). \textcolor{red}{The dynamic blockage probability in (\ref{eqn:AllBS}) can also be obtained by an extensive Markov Chain analysis for the dynamic blockage events in Appendix \ref{app:Markov_NBS}.}
% Note that the notation $P(B_d|n,\{r_i\})$ is a compressed version of $P_{B_d|N,\{R_i\}}(B_d|n,\{r_i\})$, which we use for convenience.

% where the random variables are represented in capitals and their realizations in the corresponding small letters. We keep the short notation throughout the paper for simplicity.
% where event $\mathcal{B}$ denote the blockage of all $N$ available BSs and the expression for $\alpha_i$ is taken from (\ref{eqn:SingleBS}).
% From the ~(\ref{eqn:AllBS}), we can observe that the probability of simultaneous blockage of all BSs is a function of the number of BSs ($N$) and the distance between UE and BS ($r_n$). 

% \subsection{Unconditional Blockage Probability}
To obtain the marginal LOS blockage probability $P(B^{LOS})$, we first evaluate the conditional LOS blockage probability $P(B^{LOS}|m)$ by taking the average of $P(B^{LOS}|m,\{r_i\})$ over the distribution of distances $\{r_i\}$ and then find $P(B^{LOS})$ by taking the average of $P(B^{LOS}|m)$ over the distribution of $m$.
% \begin{equation}\label{eqn:pBgivenN}
% \begin{split}
% P(B^{LOS}|m)&=\!\!\int_{r_1}\!\!\!\cdots\!\int_{r_m} \prod_{i=1}^m \;P(B_i^{LOS}|m,r_i)\times\\&\qquad\qquad\qquad f(\{r_i\}|m)\; dr_1\cdots dr_m
%  \end{split}
% \end{equation}
% \begin{equation}\label{eqn:ep1n}
% \begin{split}
%   P(B^{LOS})=\sum_{m=0}^{\infty} P(B^{LOS}|m)P_M(m).
% \end{split}
% \end{equation}

% We obtain $P(B^{LOS}|m) = (1-a)^n$ in Appendix \ref{app:th_prob}, where $a$ is defined in (\ref{eqn:a}). The final analytical expression for the marginal blockage probability is given in Theorem \ref{th1}.

\noindent
\begin{theorem} \label{th1} The  marginal LOS blockage probability and the conditional LOS blockage probability given LOS coverage (\ref{eqn:coveragep}) is given by:
% The probability of simultaneous blockage of all BSs (including out of coverage) in the disc around UE of radius $R$ as defined in (\ref{eqn:ep1n}) with the density of BSs $\lambda_{T}$ per unit area  is:

\begin{equation}\label{eqn:expblockage}
        P(B^{LOS}) = e^{-a p\lambda_T\pi R^2},
\end{equation} 
\begin{equation}\label{eqn:expblockagecond}
       P(B^{LOS}|\setC^{LOS}) = \frac{e^{-a p\lambda_T\pi R^2} - e^{-pq\lambda_T\pi R^2}}{1-e^{-pq\lambda_T\pi R^2}},
\end{equation} 
where,
\begin{equation}\label{eqn:a}
\begin{split}
%       a=\frac{2\mu}{RC}-\frac{2\mu^2}{R^2C^2} \log\left(1+\frac{RC}{\mu}\right).
a=\int_r\frac{e^{-(\beta r+\beta_0)}}{1+\frac{C}{\mu}r}\frac{2r}{R^2} dr
   \end{split}
\end{equation} 
can be computed by numerical integration.
Note that $C$ is proportional to blocker density $\lambda_B$ as shown in (\ref{eqn:C}). Also, $p$ and $q$ are defined in (\ref{eqn:p}) and are functions of the self-blockage angle $\omega$, and static blockage density $\lambda_S$, respectively.
% \begin{}\label{eqn:expblockage}
%         P(B) = e^{-\frac{2\pi R\mu\lambda_{T}}{C}}\left(1+\frac{CR}{\mu} \right)^{\frac{2\pi\mu^2\lambda_{T}}{C^2}}.
% \end{} 
\end{theorem}
\begin{proof} See Appendix \ref{app:th_prob}. 
\end{proof}
% We observed from Theorem \ref{th1} that the expected probability of simultaneous blockage decreases exponentially with the BS density $\lambda_T$. Further, note that $a\in(0,1)$, where $a\rightarrow 1$ when $RC/\mu\rightarrow 0$ and $a\rightarrow 0$ when $RC/\mu\rightarrow \infty$. Since the upper bound is trivial, we only prove the lower bound. 
We also consider an open park scenario, with no static blockages, and obtain a closed-form expression for the blockage probability considering only dynamic and self-blockage.
\begin{corollary}\label{cor:LOSprobOpen}
The coverage probability $P(\setC^{d})$ for an open park scenario is given by
\begin{equation}\label{eq:coverage_dyn}
P(\setC^{d}) = e^{-p\lambda_T\pi R^2}.
\end{equation}
The LOS dynamic blockage probability $P(B^d|\setC^{d})$, conditioned on coverage $\setC^d$, is given by
\begin{equation}\label{eqn:expblockagecond_dyn}
       P(B^d|\setC^{d}) = \frac{e^{-a' p\lambda_T\pi R^2} - e^{-p\lambda_T\pi R^2}}{1-e^{-p\lambda_T\pi R^2}},
\end{equation}
where,
\begin{equation}\label{eqn:aprime}
	\begin{split}
      a'=\frac{2\mu}{RC}-\frac{2\mu^2}{R^2C^2} \log\left(1+\frac{RC}{\mu}\right),
   \end{split}
\end{equation}
and $p=1-\omega/2\pi$.
\end{corollary}
\begin{proof}
The proof is given in Appendix \ref{app:th_prob}.
\end{proof}

The closed-form expression of blockage probability for the open park scenario gives insights while facilitating a quick analysis for network planning. 
For instance, we can approximate $a'$ (in (\ref{eqn:aprime})) by taking a series expansion of $\log(1+RC/\mu)$, i.e.,
\begin{equation}\label{eq:approx_a}
\begin{split}
a'&= \frac{2\mu}{RC}-\frac{2\mu^2}{R^2C^2} \left(\frac{RC}{\mu}-\frac{R^2C^2}{2\mu^2}+\frac{R^3C^3}{3\mu^3}+\cdots\right),\\
&\approx 1-\frac{2RC}{3\mu}, \quad \text{when } \frac{RC}{\mu} \ll 1.
% \text{ is small}.
\end{split}
\end{equation}
% thus, when $RC/\mu\rightarrow 0$, then $a\rightarrow 1$.

Note that for the blocker density as high as $0.1$ blocker/m$^2$ (bl/m$^2$), and for other parameters in Table \ref{tab:params}, we get $RC/\mu = 0.35$, which shows that the approximation (\ref{eq:approx_a}) holds for a wide range of blocker densities.
For large BS density $\lambda_T$, the coverage probability $P(\setC^{d})$ is approximately 1 and $P(B^d|\setC^{d})\approx P(B^d)$. 
In order to have a blockage probability $P(B^d)$ less than a threshold $\bar{P}$, i.e.,
\begin{equation}
P(B^d)=e^{-a'p\lambda_T\pi R^2}\leq\bar{P},
\end{equation}
the required BS density follows
\begin{equation} \label{eqn:approx_lamT}
\lambda_T\geq\frac{-\log(\bar{P})}{a'p\pi R^2}\approx \frac{-\log(\bar{P})(1+\frac{2RC}{3\mu})}{p\pi R^2},
\end{equation}
where $C$ is proportional to the blocker density $\lambda_B$ in (\ref{eqn:C}). 
% Thus, the approximation holds for smaller $\lambda_B$. 
The result (\ref{eqn:approx_lamT}) shows that the BS density approximately scales linearly with the blocker density (along with a constant factor).
\subsection{Expected Blockage Duration}
Recall that the duration of the blockage of a single BS-UE link is an exponential random variable with mean $1/\mu$. 
% $T_i \sim \exp(\mu)$. %, i.e.,
%\begin{equation}
%f_{T_i}(t_i) = \mu e^{-\mu t_i},\quad \text{for} \;\;i=1:n.
%\end{equation}
Since the blockage of individual BS-UE links are independent, the duration $T^{LOS}$ of the LOS blockage of all $n$ BSs (which are in coverage) follows an exponential distribution with mean $1/n\mu$ as follow
%We show that the duration of the blockage of all $n$ BSs follows an exponential distribution with mean $1/n\mu$.   
%Consider a time instant when all $n$ BSs are blocked; the residual duration of the blockage period of the $i^{\text{th}}$ BS-UE link follows the same distribution as $f_{T_i}(t_i)$ because of the memoryless property of the exponential distribution.
%Therefore, the duration of the period of simultaneous blockage of all $n$ BSs is a random variable $T_B=\min\{T_1,T_2,\cdots,T_n\}$. Note that $T_B$ follows the distribution 
%$T_B \sim \exp(n\mu)$, conditioned on the number of BSs $N=n$. 
% We can therefore write the expected LOS blockage duration as 
\begin{equation}\label{eqn:exp_TBgivenN}
\E[T^{LOS}|n] = \frac{1}{n\mu},
\end{equation}
where $T^{LOS}$ denote the LOS blockage duration.  
% and,
% \begin{equation}
% \E[D_B|N] = \frac{1}{n\mu}P(B|N)
% \end{equation}
% becomes exponential random variable $(D_B \sim \exp(N\mu))$, i.e., $min\{D_1,D_2,\cdots,D_N\}$. As, the BSs are considered as homogeneous PPP, we need to sum over number of BSs. Thus, the expected duration of blockage is evaluated as
% \begin{equation}\label{eqn:expdurblck}
%         \mathbb{E} \left[D_B|\setC^{LOS}\right] := \frac{\sum_{n=1}^N \frac{1}{n\mu} P_N(n)}{1-\mathbb{P}(n=0)}.
% \end{equation}

\noindent
\begin{theorem}\label{th3} The expected blockage duration of simultaneous blockage of all the BSs in $B(o,R)$, conditioned on the coverage event $\setC^{LOS}$ in (\ref{eqn:coveragep}), is obtained as
\begin{equation}\label{eqn:blockageDurationExp}
        \mathbb{E}\left[T^{LOS}|\setC^{LOS}\right] = \frac{e^{-pq\lambda_T\pi R^2}}{\mu\left(1-e^{-pq\lambda_T\pi R^2}\right)}\text{Ei}\left[pq\lambda_T\pi R^2\right],
\end{equation}        
where, $\text{E}\text{i}\left[pq\lambda_T\pi R^2\right]$
% = \int_{0}^{p\lambda_T\pi R^2} \frac{e^x-1}{x} dx$ %Rajeev: Is this important??
= $\sum_{n=1}^\infty\frac{[pq\lambda_T\pi R^2]^n}{nn!}$ is a well-known series which can be written as the exponential integral function~\cite{weisstein2002exponential}.
\end{theorem}
\begin{proof}
 See Appendix \ref{app:th_dur}.
\end{proof}
\noindent
\begin{lemma}
$\text{E}\text{i}\left[pq\lambda_T\pi R^2\right]$ converges.
\end{lemma} 
\begin{proof}
 We can use Cauchy ratio test to show that the series $\sum_{n=1}^\infty\frac{[pq\lambda_T\pi R^2]^n}{nn!}$ is convergent. Consider 
$L = \lim_{n\rightarrow\infty}\frac{[pq\lambda_T\pi R^2]^{n+1}/((n+1)(n+1)!)}{[pq\lambda_T\pi R^2]^n/(nn!)} = \lim_{n\rightarrow\infty}\frac{[pq\lambda_T\pi R^2]n}{(n+1)^2}=0$. Hence, the series converges.
\end{proof}

An approximation of blockage duration (\ref{eqn:blockageDurationExp}) can be obtained for a high BS density ($\lambda_T$) as follows:
\begin{equation}\label{eqn:durLOSapproxfinal}
\begin{split}
&\E\left[T^{LOS}|\setC^{LOS}\right]
% &\approx \frac{1}{\mu\left(1-e^{-pq\lambda_T\pi R^2}\right)}\left(\frac{1}{ pq\lambda_T\pi R^2}+ \frac{1}{( pq\lambda_T\pi R^2)^2}\right)\\
\approx\frac{1}{\mu pq\lambda_T\pi R^2\left(1-e^{-pq\lambda_T\pi R^2}\right)}.
\end{split}
\end{equation}

This approximation is justified in Appendix \ref{app:approx_dur}.

\begin{corollary}
The expected blockage duration given coverage $\setC^d$ (\ref{eq:coverage_dyn}) for the open park scenario is given by
\begin{equation}
\mathbb{E}\left[T^{d}|\setC^{d}\right] = \frac{e^{-p\lambda_T\pi R^2}}{\mu\left(1-e^{-p\lambda_T\pi R^2}\right)}\text{Ei}\left[p\lambda_T\pi R^2\right],
\end{equation}
which can be obtained from (\ref{eqn:blockageDurationExp}), by setting $q=1$ (corresponding to $\beta=0$ and $\beta_0=0$).
\end{corollary}

\subsection{Expected Blockage Frequency}
For the open park scenario, we define the dynamic blockage frequency $\zeta^d$ as the rate of blockage of all the BSs in the range of UE, i.e.,
\begin{equation}\label{eqn:AllBSfreq}
\begin{split}
    \zeta^d = n\mu P(B^{d}|n,\{r_i\}),
%     =n\mu \prod_{i=1}^n\frac{\frac{C}{\mu}r_i}{1+\frac{C}{\mu}r_i},
\end{split}
\end{equation}
where $P(B^{d}|n,\{r_i\})$ represents the dynamic blockage probability given $n$ and $\{r_i\}$, $n$ is the number of BS in the disc $B(o,R)$ not blocked by the user's body, and $\{r_i\}$ is the set of UE-BS distances $r_i$, for $i=1,\cdots,n$. We do not consider the effect of static blockages for the open park scenario. A closed-form expression for the expected blockage frequency is obtained in Theorem \ref{th2}.
% The expected rate of blockage is obtained where the expectation is taken over the joint distribution of $N$ and $\{R_i\}$.
% \begin{equation}\label{eqn:blockageDurEQn1}
% \E[\zeta_B|N] = \!\!\int\!\!\!\int_{r_i}\!\! \zeta_B\; f(\{R_i\}|N)\;dr_1\cdots dr_n,
% \end{equation}
% % following from (\ref{eqn:blockageDurEQn1}), we have
% \begin{equation}\label{ep1n}
%         \mathbb{E} \left[\zeta_B\right] = \sum_{n=0}^\infty \E[\zeta_B|N]P_N(n).
% \end{equation}
\begin{theorem} \label{th2} The expected blockage frequency given coverage $\setC^d$ (\ref{eq:coverage_dyn}) for the open park scenario is obtained as:
% \begin{equation}\label{eqn:blkfrq}
%         \mathbb{E}\left[\zeta_B\right] = \mu (1-a')p \lambda_T\pi R^2 e^{-a'p\lambda_T\pi R^2},
% \end{equation}       
% and the expected frequency conditioned on the coverage event (\ref{eqn:C}) is 
\begin{equation}\label{eqn:blkfreq_cond}
\begin{split}
\MoveEqLeft \mathbb{E}\left[\zeta^d|\setC^{d}\right] = \frac{\mu (1-a')p\lambda_T\pi R^2e^{-a'p\lambda_T\pi R^2}}{{1-e^{-p\lambda_T\pi R^2}}}.
% \frac{1}{{1-e^{-\lambda_T\pi R^2}}} \times \\
% \MoveEqLeft \left[\mu a \lambda_T\pi R^2 e^{-\frac{2\pi R\mu\lambda_{T}}{C}}\left(1+\frac{CR}{\mu} \right)^{\frac{2\pi\mu^2\lambda_{T}}{C^2}}-e^{-\lambda_T\pi R^2}\right].
\end{split}
\end{equation}
where $p$ and $a'$ are defined in (\ref{eqn:self}) and (\ref{eqn:aprime}) respectively.
\end{theorem}
\begin{proof}See Appendix \ref{app:th_freq}.
\end{proof}

\section{Effect of NLOS links on Blockage}\label{sec:NLOS}
In the previous section, we analyzed the LOS blockage probability and duration of blockage events due to mobile blockers. Apart from the LOS signal, the NLOS paths through reflections by buildings, trees and other permanent structures (collectively called reflectors) also play a major role in the mmWave cellular systems. 
% Akdeniz \textit{et al.}~\cite{CellularCap-Rap} have shown through measurements that there would be at least one NLOS path available for each BS-UE link. We therefore extend our blockage model to the NLOS links for a complete blockage analysis of mmWave cellular system. 
We first present a NLOS link model and evaluate the coverage probability by incorporating the NLOS links. Then we will evaluate a lower bound on the blockage probability given coverage by considering both LOS and NLOS links.
\subsection{Coverage Probability incorporating NLOS links}
When we consider both LOS and NLOS links as potential links for connectivity, the coverage is defined as the availability of at least one unblocked LOS or NLOS link (unblocked by static or self-blockage). Considering the experimental results by Akdeniz \textit{et al.}~\cite{CellularCap-Rap}, there is always at least one NLOS signal for a given BS location, and we assume the NLOS signal is sufficiently strong when the BS is in a disc $B(o,\Rt)$ around the UE (See Figure \ref{fig:NLOS_analysis_approx}). Therefore, we say that an NLOS link is not available when $r_i> \Rt$.
% which we indicate by $I_{(r_i> \Rt)}$. 
Assuming the availability of LOS and NLOS links to be independent of each other, we obtain the coverage probability $P(\setC|m,\{r_i\})$ conditioned on $m$ and $\{r_i\}$ for $i=1,\ldots, m$ as follow
\begin{equation}
P(\setC|m,\{r_i\}) = 1-\prod_{i=1}^m\left[\left(1-pe^{-(\beta r_i+\beta_0)}\right)\left(1-I_{(r_i\le \Rt)}\right)\right],
\end{equation}
where $I_{(r_i\le \Rt)}=1$ when $r_i\le \Rt$ and $I_{(r_i\le \Rt)}=0$ otherwise.
\begin{lemma}\label{lemma:NLOScoverage}
    The coverage probability considering LOS and NLOS links is defined by
    \begin{equation}\label{eqn:coverageNLOS}
        P(\setC) = 1-e^{-\qt\lambda_T\pi R^2},
    \end{equation}
    where,
    \begin{equation}\label{eqn:qt}     
    \qt = \frac{\Rt^2}{R^2}-\frac{2pe^{-\beta_0}}{\beta^2\Rt^2}\left((1+\beta R)e^{-\beta R}-(1+\beta \Rt)e^{-\beta \Rt}\right).
    \end{equation}
\end{lemma}

% o NLOS links if at least one NLOS link is available to the UE. \\
\begin{proof}
The proof is provided in Appendix \ref{app:proof_NLOS_coverage}.
\end{proof}

\subsection{Incorporating NLOS links in Blockage Probability}
The NLOS links can potentially help in achieving high coverage and reducing the blockage probability. In general, for a given BS at a distance $r_i$ from the origin, the corresponding $k$ virtual BSs are located at a distance larger than $r_i$. The corresponding NLOS link length can be obtained by tracing the path from the BS to a reflection point (the point where the NLOS signal got reflected), and then from the reflection point to the UE. However, in order to avoid the underlying complexity, we present a lower bound on blockages for the NLOS case by considering that all the $k$ virtual BSs are located at a distance $r_i$ from the origin, and uniformly over an angular range of $[0,2\pi]$. By thus minimizing the length of NLOS paths and making their angular range independent of the BS, we will minimize the incidence of blockage. For such a scenario, the NLOS blockage probability $P(B_i^{NLOS}|m,r_i,k)$, corresponding to the $i$th BS, considering blockage by dynamic blockers, is given by:
\begin{equation}\label{eqn:NLOSprobBound1}
\begin{split}
P(B_i^{NLOS}|m,r_i,k) 
&= \prod_{j=1}^k\left(\frac{\frac{C}{\mu}r_i}{1+\frac{C}{\mu}r_i}I_{(r_i\le \Rt)}+1-I_{(r_i\le \Rt)}\right)\\
% &= \prod_{j=1}^k\left(1-\frac{1}{1+\frac{C}{\mu}r_i}I_{(r_i\le \Rt)}\right)\\
&=\left(1-\frac{1}{1+\frac{C}{\mu}r_i}I_{(r_i\le \Rt)}\right)^k\\
&=(1-\bt(r_i))^k,
\end{split}
\end{equation}
where we assume the independent blocking of all the NLOS links and define 
\begin{equation}\label{eqn:btilda_r}
\bt(r_i)=\frac{1}{1+\frac{C}{\mu}r_i}I_{(r_i\le \Rt)}.
\end{equation}

Next, we integrate the blockage probability in (\ref{eqn:NLOSprobBound1}) w.r.t. the distribution of $k$ given in (\ref{eqn:K}) as follows:
\begin{equation}\label{eqn:NLOSprobBound2}
\begin{split}
P(B_i^{NLOS}|m,r_i)& = \sum_{k=1}^\infty P(B_i^{NLOS}|m,r_i,k)P_K(k) \\
&=\sum_{k=1}^{\infty}(1-\bt(r_i))^kP_K(k).
\end{split}
\end{equation}

We further solve (\ref{eqn:NLOSprobBound2}) and write the final expression as follows (we omit the calculation for brevity):
 \begin{equation}\label{eqn:NLOSprobBound3}
\begin{split}
P(B_i^{NLOS}|m,r_i)= e^{-\bt(r_i)\kappa}-\bt(r_i) e^{-\kappa}.
\end{split}
\end{equation}

 The blockage probability $P(B)$ follows by assuming independent blocking of LOS and NLOS links and by averaging over the distributions  $P_M(m)$ and $f(\{r_i\}|m)$, given in (\ref{eqn:poisson}) and (\ref{eqn:distribution}),  respectively, as follows:
\begin{equation}\label{eqn:probBLbothLOSandNLOS}
\begin{split}
P(B) &= \sum_{m=0}^{\infty}\int_{r_1}\!\!\!\cdots\!\!\!\int_{r_m}\prod_{i=1}^{m}P(B_i^{LOS}|m,r_i)P(B_i^{NLOS}|m,r_i)\\&\qquad\qquad\qquad\quad \times f(\{r_i\}|m)\,d r_1\cdots dr_m\, P_M(m).
\end{split}
\end{equation}

The final expression of the blockage probability $P(B)$ is given in Theorem \ref{th:LOSandNLOS}.
\begin{theorem}\label{th:LOSandNLOS}
The blockage probability $P(B)$ considering both LOS and NLOS links is given by,
\begin{equation}
\begin{split}
P(B) = e^{-\at\lambda_T\pi R^2},
\end{split}
\end{equation}
where $\at$ is obtained as,
\begin{equation}\label{eqn:at}
\begin{split}
\at &= 1-\int_{r}\left(1-\frac{pe^{-(\beta r +\beta_0)}}{1+\frac{C}{\mu}r}\right)\left(e^{-\bt(r)\kappa}-\bt(r) e^{-\kappa}\right)\frac{2r}{R^2}dr,
\end{split}
\end{equation}
where $\bt(r)$ is given in (\ref{eqn:btilda_r}).

The conditional probability given coverage $\setC$ (\ref{eqn:coverageNLOS}) is given by
\begin{equation}
P(B|\setC) =  \frac{e^{-\at\lambda_T\pi R^2}-e^{-\qt\lambda_T\pi R^2}}{1-e^{-\qt\lambda_T\pi R^2}},
\end{equation}
where $\qt$ and $\at$ are given in (\ref{eqn:qt}) and  (\ref{eqn:at}) respectively.
\end{theorem}
The proof follows by putting the expressions (\ref{eqn:blProbLOSstep1}) and (\ref{eqn:NLOSprobBound3}) in (\ref{eqn:probBLbothLOSandNLOS}) and following steps similar to the proof of Theorem \ref{th1} in Appendix \ref{app:th_prob}.

\subsection{Expected Blockage Duration considering NLOS links}
Similar to the LOS case, we can evaluate the the expected duration of blockage events assuming independent blocking of LOS and NLOS links. We define the number of BS in the disc $B(o,\Rt)$ as $\nt$, which follows $\text{Poisson}(\lambda_T\pi \Rt^2)$ as
\begin{equation}
f_{\tilde{N}}(\nt) = \frac{[\lambda_T\pi \Rt^2]^{\nt}}{\nt!}e^{-\lambda_T\pi \Rt^2}.
\end{equation}

Given a BS in $B(o,\Rt)$, there are $k$ virtual BSs corresponding to that actual BS. Therefore, the total number of NLOS links are $\nt k$. We know from Section \ref{sec:analysis} that  $n$ LOS paths cover the UE. Therefore, the total number of LOS and NLOS paths is $n+\nt k$. The expected blockage duration is defined as
\begin{equation}\label{eqn:durNLOS}
\E[T|\setC] = \frac{1}{1-e^{-\qt\lambda_T\pi R^2}}\E\left[\frac{1}{\mu(n+\nt k)}\right], 
\end{equation}
where the expectation is taken over the distribution of $n$, $\nt$, and $k$. We can approximate the blockage duration in (\ref{eqn:durNLOS}) using a first order approximation as follows:
\begin{equation}\label{eqn:durNLOSapprox}
\begin{split}
\E[T|\setC] &\approx \frac{1}{1-e^{-\qt\lambda_T\pi R^2}}\frac{1}{\E[\mu(n+\nt k)]}\\
&=\frac{1}{1-e^{-\qt\lambda_T\pi R^2}}\frac{1}{\mu\left(\E[n]+\E[\nt]\E[ k]\right)}\\
&=\frac{1}{1-e^{-\qt\lambda_T\pi R^2}}\frac{1}{\mu\left(pq\lambda_T\pi R^2 + \kappa\lambda_T\pi \Rt^2\right)},
\end{split}
\end{equation}
where $p$ and $q$ are defined in (\ref{eqn:p}) and $\qt$ defined in (\ref{eqn:qt}). Equation  (\ref{eqn:durNLOSapprox}) is a good approximation of blockage duration for high BS densities, as discussed in Appendix \ref{app:approx_dur}.

\begin{figure}[!t]
    \centering
    \includegraphics[width=0.4\textwidth]{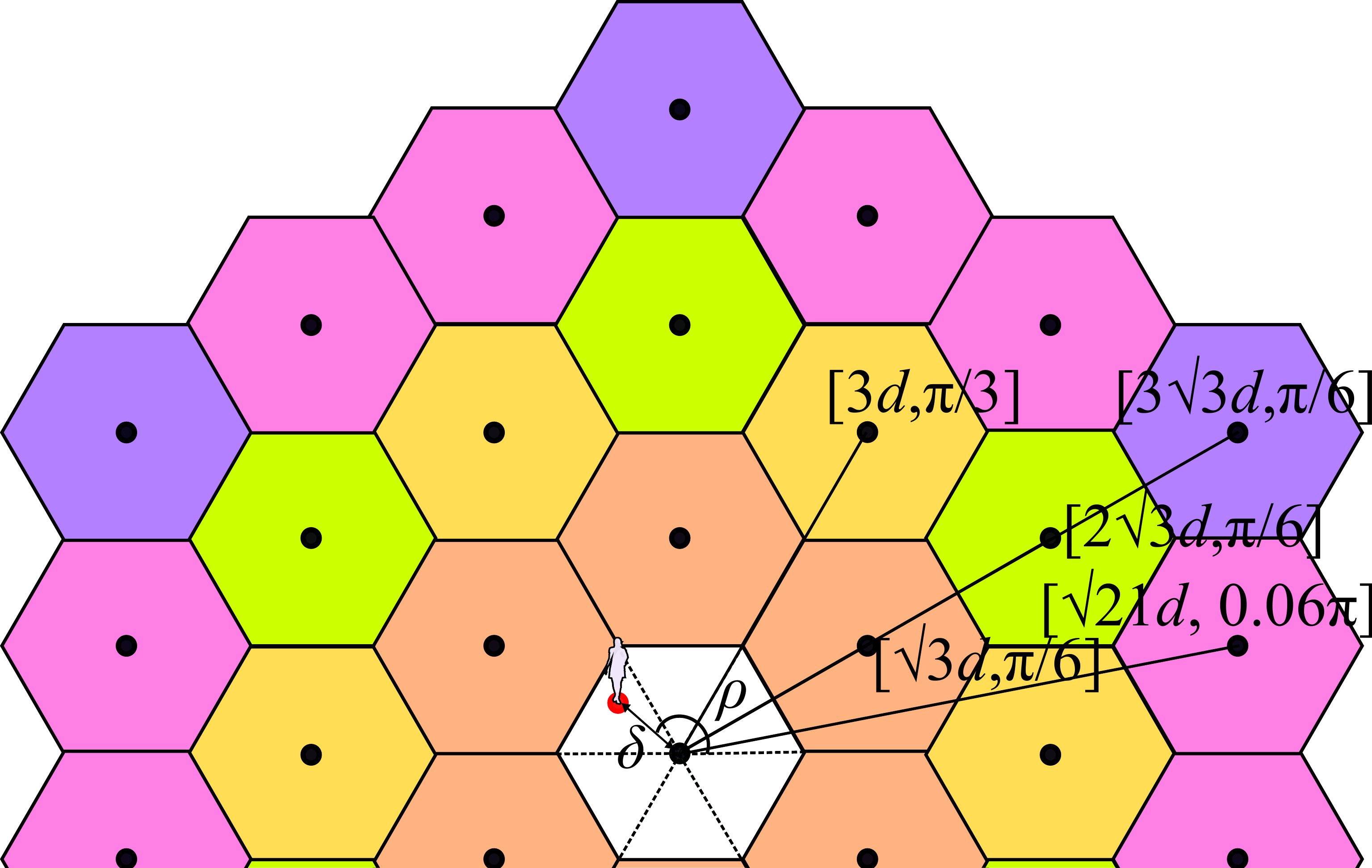}
    \caption{Hexagonal cellular layout.}
    \label{fig:hexcell}
    \vspace{-4mm}
\end{figure}

\begin{table}[!t]
\vspace{2mm}
\caption{Hexagonal cell distances and angular distributions.}
\label{tab:hexcell}
\centering
\begin{tabular}{|c|c|c|c|}\hline
    \textbf{Level} & \textbf{distance $d_i$} & \textbf{angle $\psi_i$}&\textbf{Num BS}\\\hline
    0&0&reference&1\\
    1&$\sqrt{3}d$&$[\pm\pi/6,\pm\pi/2,\pm5\pi/6]$ &6\\
    2&3d&$[0,\pm\pi/3,\pm2\pi/3,\pi]$&6\\
    3&$2\sqrt{3}d$&$[\pm\pi/6,\pm\pi/2,\pm5\pi/6]$&6\\
    4&$\sqrt{21}d$&$[\pm 0.06\pi,\pm\pi/3\pm 0.06\pi,$&12\\
    &&$\pm2\pi/3\pm 0.06\pi,\pm(\pi-0.06\pi)]$&\\
    5&3$\sqrt{3}d$&$[\pm\pi/6,\pm\pi/2,\pm5\pi/6]$&6\\
\hline
\end{tabular}
% \vspace{-5mm}
\end{table}

%!!!!!!!!!!!!!!!!!!!!!!!!!!!!!!!!
\section{Hexagonal Cell Layout}
\label{sec:hex}
In Section \ref{sec:analysis}, we analyzed the random deployment of BS in an area. We observed that a high density of BS is required to meet the QoS requirements (See  Section \ref{sec:numResults}). These results motivated us to look at a planned network such as a grid-based hexagonal cell layout for the open park scenario.  For such a planned network, we analyze the probability of blockage events due to random mobile blockers. We consider the hexagonal cell layout shown in Figure \ref{fig:hexcell}. A typical BS is located at the origin, and the UE is located within that cell at a distance $\delta$ from the origin and angle $\rho$ from the x-axis. The user location is uniform in the central hexagonal cell. The other cells follow the hexagonal grid around the central cell. 
The distances  $d_i$ from the origin and the angles $\psi_i$ from the x-axis of BSs at different cells are given in Table \ref{tab:hexcell} and shown in Figure \ref{fig:hexcell}. We define a set of BSs lie at a particular level $\ell$ (shown in Table \ref{tab:hexcell}, and shown as hexagons with the same color in Figure \ref{fig:hexcell}), if they all are at the same distance from the origin.
The distance $r_i$ of the UE from the $i$th BS is given by:
\begin{equation}\label{eqn:hexDiatance}
r_{i}^2 = d_{i}^2+\delta^2-2d_{i}\delta\cos(\psi_{i}-\rho), 
\end{equation}
which is a function of the UE's random location $\{\delta,\rho\}$.
% We assume that the UE is located uniformly in cell. What is the distribution of $r,\vartheta$? Also, find the distribution of $d_{ik}$?

For self-blockage, we consider a circle around UE of radius $R$ and define a sector of this circle with an angle $\omega$ as the self-blockage zone. We consider the worst case of self-blockage by choosing a location and orientation of the user which accounts for the blockage of the maximum number of BSs. For instance, for $\omega=60^{\text{o}}$, there can be up to 10 BSs blocked by self-blockage out of total 37 BSs (See Table \ref{tab:hexcell} and Figure \ref{fig:hexcell}). Considering this worst-case scenario, we get an upper bound on the blockage probability when we consider self-blockage.
% Find which BSs are blocked as a function of the user's orientation $\vartheta$? 

% For now we assume 

\begin{table}[!t]
%\vspace{2mm}
\caption{Simulation parameters}
\label{tab:params}
\centering
\begin{tabular}{|c|c|}\hline
	\textbf{Parameters} & \textbf{Values} \\\hline 
    % LOS Range,  $R$ &100 m \\
    % NLOS Range,  $\Rt$ & 65 m \\
    %Density of Blockers, $\lambda_B$ & 0.1/m$^2$, 0.2/m$^2$\\
    Velocity of dynamic blockers, $V$&1 m/s\\
    Height of dynamic blockers, $h_B$ & 1.8 m\\
    Height of UE, $h_R$ & 1.4 m\\
    Height of BSs, $h_T$ & 5 m\\
    Expected blockage duration, $1/\mu$ &1/2 s\\
   Self-blockage angle, $\omega$ & 60$^\text{o}$ \\ 
   Parameter for the number of NLOS paths, $\kappa$ & $3$ \\
   Average size of static blockages, $\ell\times w$ &$10$ m $\times 10$ m\\
\hline
\end{tabular}
% \vspace{-5mm}
\end{table}
 
The expression for dynamic blockage probability for the hexagonal case can be obtained by using the expression in (\ref{eqn:AllBSdynamic}) and assuming the independent blocking of all the links as follows: 
\begin{equation}\label{eqn:AllBS_hex}
        P(B^{\text{hex}}|\delta,\rho)  =\prod_{i=1}^{m'}\left(1-\frac{1}{1+\frac{C}{\mu}r_i}\right),
\end{equation} 
where $m'$ is the number of BSs which are within the range of the UE and are not blocked by self-blockage, and $r_i$ for $i=1,\cdots,m'$ is a function of $(\delta,\rho)$, as given in (\ref{eqn:hexDiatance}).
% where $d_i$ is a function of $(r,\theta)$. 
We perform a numerical integration over the distribution of the UE's location $\{\delta,\rho\}$, assuming the UE is uniformly located in the central hexagon.
%--------------------figure Tables

\begin{table}[!t]
%\vspace{2mm}
\caption{Reliability and latency requirements~\cite{DBLP,orlosky2017virtual}} 
\label{tab:applications}
\centering
\begin{tabular}{|c|c|c|c|}\hline
	\textbf{Application} & \textbf{Reliability [\%]} & \textbf{Latency [ms]} & \textbf{Caching}\\\hline 
	V2X & $\ge$ 99  & $\le$ 20 & $\mathbf{\times}$ \\
    AR/VR & $\ge$ 99.99 & $\le$ 20 & $\checkmark$ \\
    Smart Grid & $\ge$ 99.999 & $\le$ 10 & $\mathbf{\times}$ \\
    Tactile Internet & $\ge$ 99.999 & $\le$ 10 & $\mathbf{\times}$ \\
\hline
\end{tabular}
\vspace{-4mm}
\end{table}

\begin{figure}[!t]
	\centering
   \begin{subfigure}[b]{0.20\textwidth}
% 		\captionsetup{skip=-1pt}	
       \includegraphics[width=\textwidth]{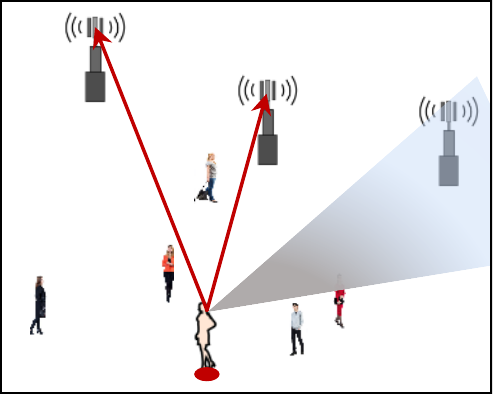}
       \caption{\small{Open park scenario.}}
       \label{fig:open}
	\end{subfigure}%
   \begin{subfigure}[b]{0.20\textwidth}
% 		\captionsetup{skip=-1pt}	
       \includegraphics[width=\textwidth]{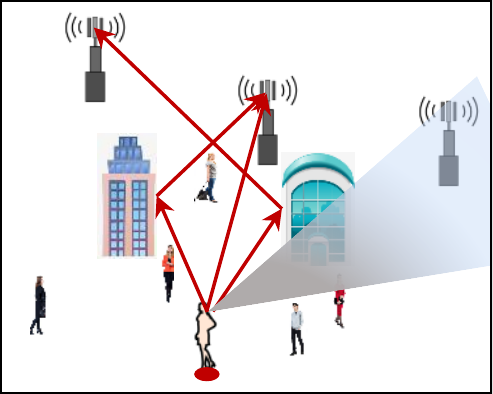}
       \caption{\small{Urban scenario.}}
       \label{fig:urban}
	\end{subfigure}%
\caption{Outdoor scenarios considered for numerical analysis. We represent static blockages by buildings, dynamic blockers by moving pedestrians, and self-blockage by a cone centered around the UE.}
\label{fig:scenarios}
\vspace{-4mm}
\end{figure}

\begin{figure*}[!t]
	\centering
    %\begin{subfigure}[b]{0.33\textwidth}
    %	\captionsetup{skip=-2pt}	
	%	\includegraphics[width=\textwidth]{img/coverageLOS}
     %   \caption{\small{Coverage probability}}
     %   \label{fig:condBlProb}
	%\end{subfigure}%
    \begin{subfigure}[b]{0.33\textwidth}
    	\captionsetup{skip=-2pt}	
		\includegraphics[width=\textwidth]{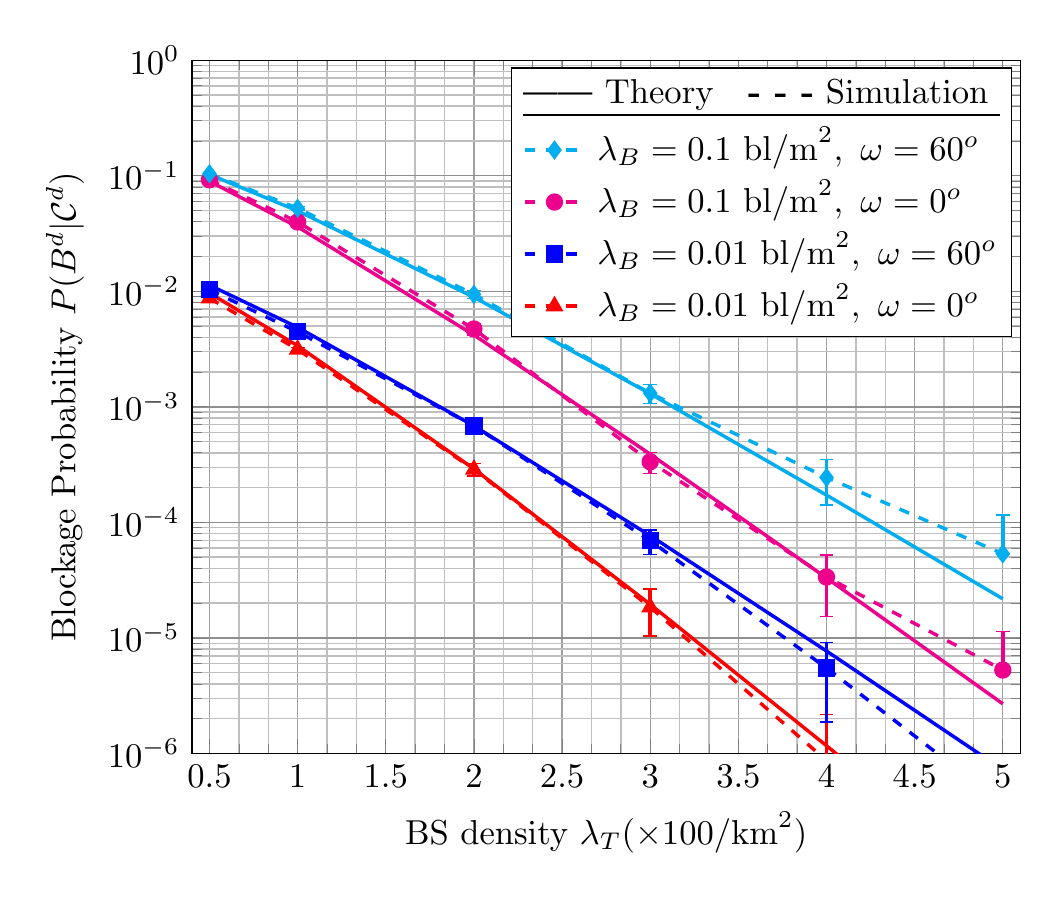}
        \caption{\small{Dynamic blockage probability.}}
        \label{fig:condBlprob}
	\end{subfigure}%
       \begin{subfigure}[b]{0.33\textwidth}
		\captionsetup{skip=-2pt}	
       \includegraphics[width=\textwidth]{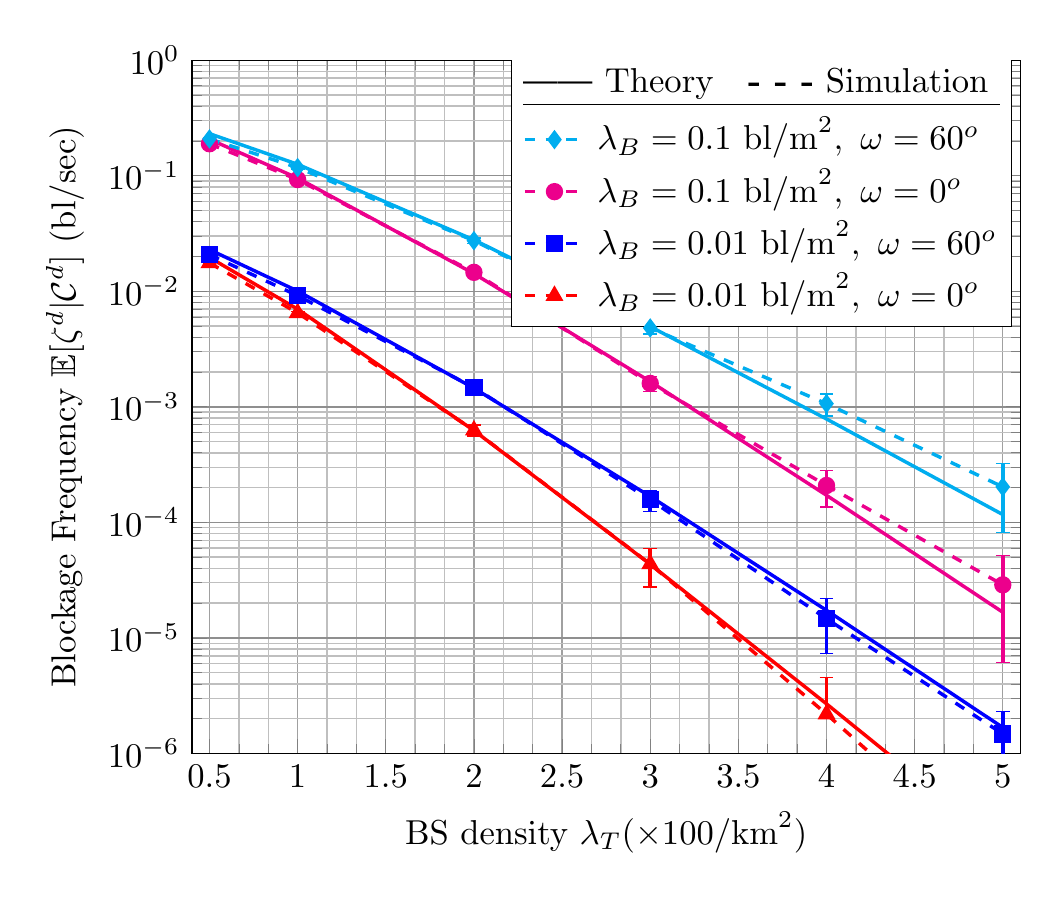}
       \caption{\small{Expected dynamic blockage frequency.}}
       \label{fig:condBlfreq}
	\end{subfigure}%
    \begin{subfigure}[b]{0.315\textwidth}
    	\captionsetup{skip=-2pt}	
		\includegraphics[width=\textwidth]{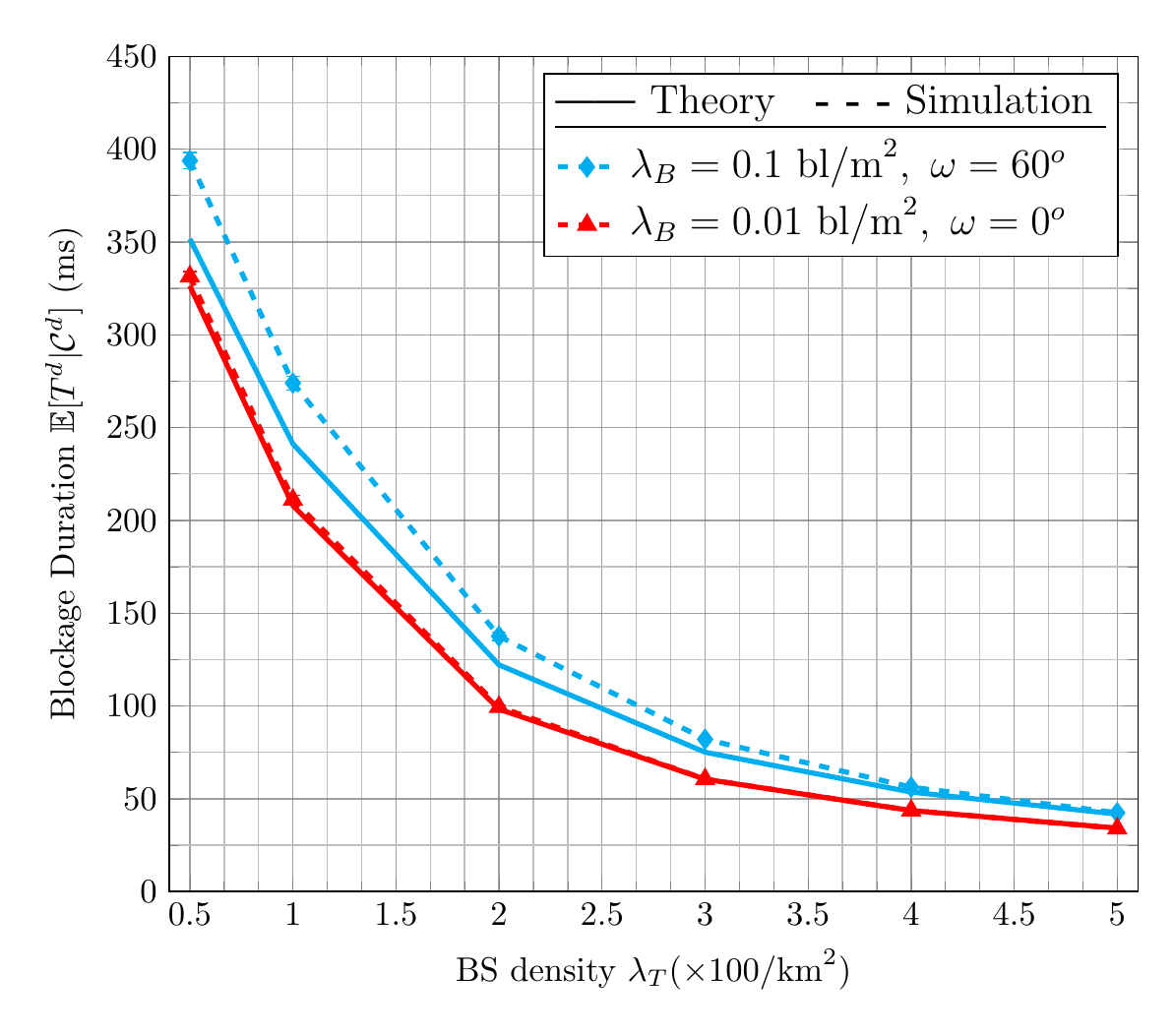}
		\caption{\small{Expected dynamic blockage duration.}}
		\label{fig:condBldur}
	\end{subfigure}%
\caption{Open park scenario: Conditional dynamic blockage statistics given coverage (defined by at least one BS within reach of the UE and outside the self-blockage zone) for a communication range $R=100$ m. The conditional probability and duration of dynamic blockage events are shown against BS density $\lambda_T$ for various blocker densities $\lambda_B$ and self-blockage angles $\omega$. 
}
\label{fig:condBl}
\vspace{-4mm}
\end{figure*}

\begin{figure*}[!t]
	\centering
    %\begin{subfigure}[b]{0.33\textwidth}
    %	\captionsetup{skip=-2pt}	
	%	\includegraphics[width=\textwidth]{img/coverageLOS}
     %   \caption{\small{Coverage probability}}
     %   \label{fig:condBlProb}
	%\end{subfigure}%
    \begin{subfigure}[b]{0.33\textwidth}
    	\captionsetup{skip=-2pt}	
		\includegraphics[width=\textwidth]{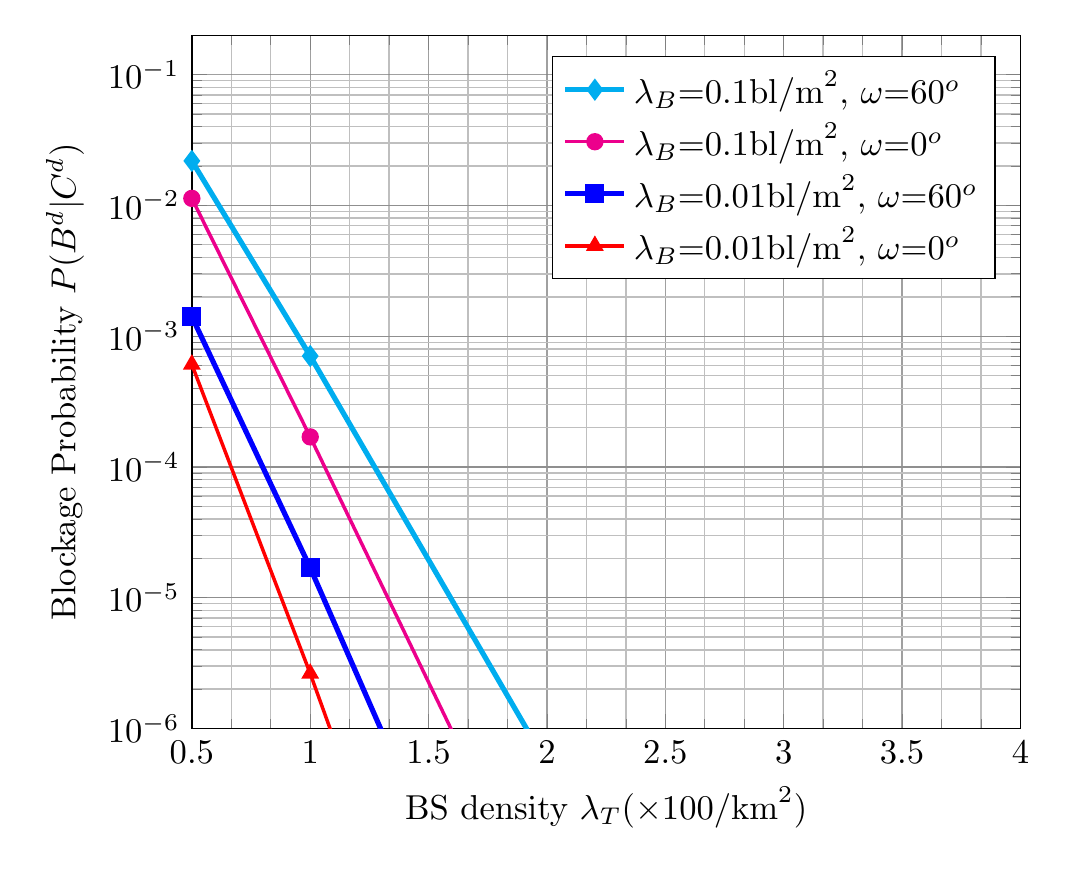}
        \caption{\small{Dynamic blockage probability.}}
        \label{fig:condBlprobR200}
	\end{subfigure}%
       \begin{subfigure}[b]{0.33\textwidth}
		\captionsetup{skip=-2pt}	
       \includegraphics[width=\textwidth]{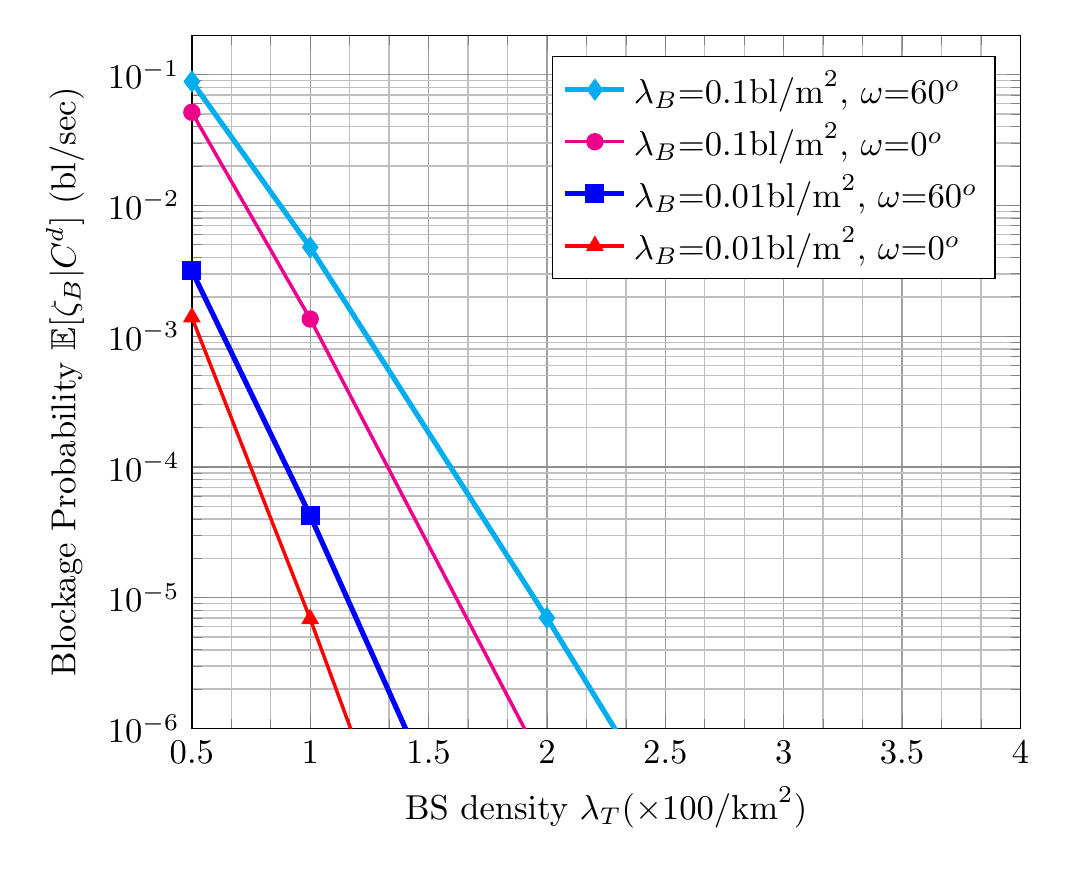}
       \caption{\small{Expected dynamic blockage frequency.}}
       \label{fig:condBlfreqR200}
	\end{subfigure}%
    \begin{subfigure}[b]{0.315\textwidth}
    	\captionsetup{skip=-2pt}	
		\includegraphics[width=\textwidth]{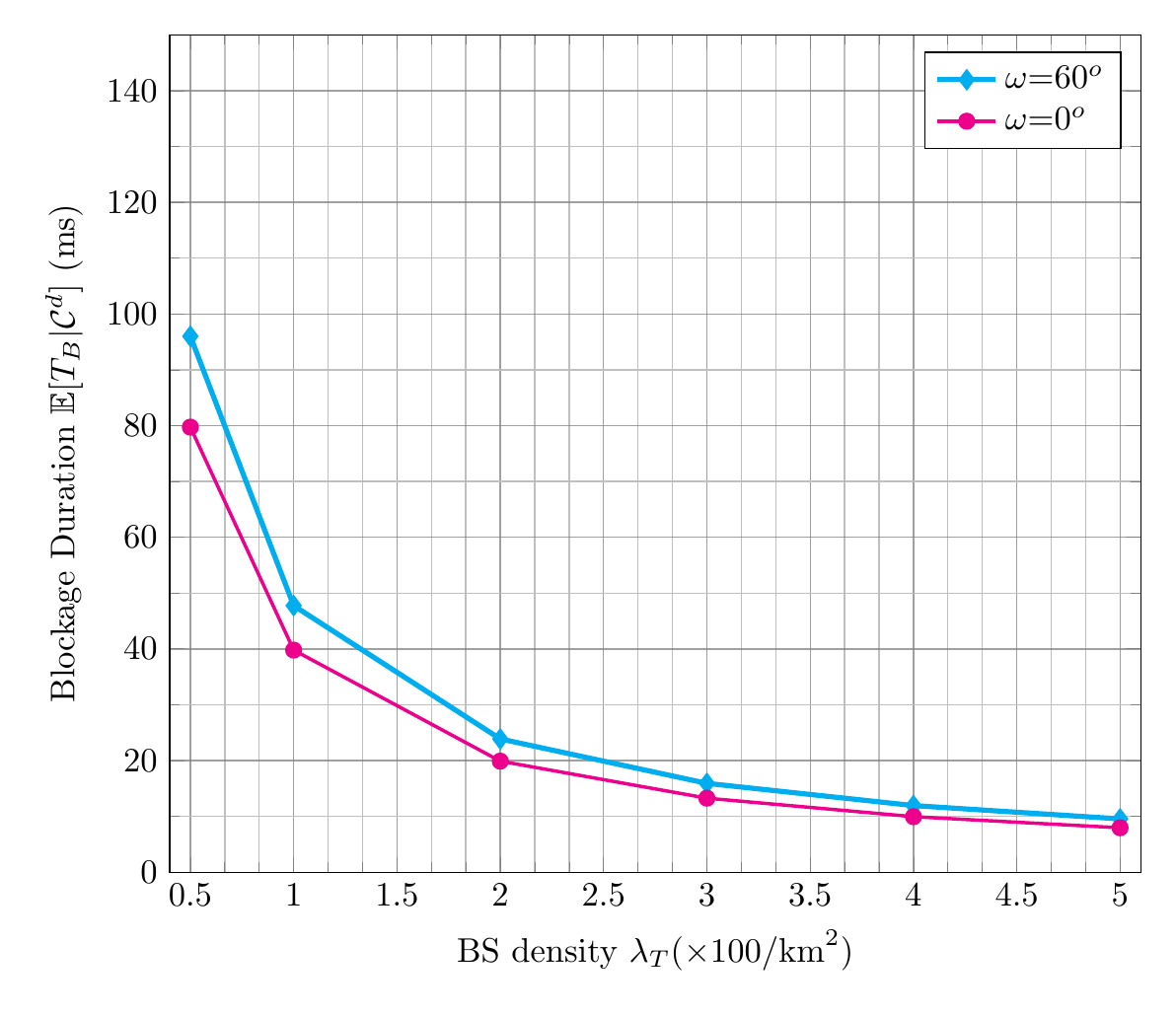}
		\caption{\small{Expected dynamic blockage duration.}}
		\label{fig:condBldurR200}
	\end{subfigure}%
\caption{Open park scenario: Conditional dynamic blockage statistics given coverage for a communication range $R=200$ m.}
\label{fig:condBlR200}
\vspace{-4mm}
\end{figure*}

% \begin{figure}[!t]
% 	\centering
% 	\includegraphics[width=0.435\textwidth]{img/prob_blockage_cond}
%     \captionsetup{skip=-2pt}
% 	\caption{Blockage Probability conditioned on coverage event vs BS density for different values of blocker density and self-blockage angle.}
% 	\label{fig:condBlProb}
% 	\vspace{-6mm}
% \end{figure}
% \begin{figure}[!t]
% 	\centering
% 	\includegraphics[width=0.435\textwidth]{img/blockage_freq_cond}
%     \captionsetup{skip=-2pt}
% 	\caption{Expected blockage frequency conditioned on coverage event vs BS density.}
% 	\label{fig:condBlfrac}
% 	\vspace{-5mm}
% \end{figure}
% \begin{figure}[!t]
% 	\centering
% 	\includegraphics[width=0.42\textwidth]{img/blockage_time}
%     \captionsetup{skip=-2pt}
% 	\caption{Expected blockage duration conditioned on coverage event vs BS density. Note that the theoretical blockage duration is the same for different blocker densities for a fixed self-blockage angle $\omega$.}
% 	\label{fig:condBldur}
% 	\vspace{-5mm}
% \end{figure}

%--------------Results-----------------
\section{Numerical Evaluation}
\label{sec:numResults}
We consider two outdoor scenarios for 5G mmWave cellular networks as shown in Figure \ref{fig:scenarios}.
%: (i) open park scenario: lack of permanent structure, which results in low static blockage and reflecting paths, and (ii) urban scenario: large number of static blockers, which results into a higher static blockage and reflecting path. , due to lack of permanent structures and reflecting paths. Further, indoor environment and cities are considered under urban scenario, due to high number of permanent structures and reflecting paths.     
\begin{enumerate}
\item \textit{Open park scenario}: In an open park scenario (shown in Figure \ref{fig:scenarios}(\subref{fig:open})), due to lack of buildings and permanent structures, we assume that the UE does not suffer static blockages. We also assume that due to the lack of reflecting surfaces, we may not have strong NLOS paths available. Other environments, such as those found in rural areas, may also fall in this category. We only considered dynamic and self-blockage of LOS links in this scenario. 
% \textcolor{red}{We consider applications such as AR/VR in parks, V2X, and car entertainment services that may include AR/VR streaming for this scenario.}
% LOS paths for the evaluation of performance. The LOS path between UE and BSs in this scenario can still suffer dynamic and self-blockage. 
%Further, applications that may be used by users mainly consists of AR/VR video streaming and online gaming.
\item \textit{Urban scenario}: In the case of an urban scenario (shown in Figure \ref{fig:scenarios}(\subref{fig:urban})), the signal to the UE may suffer static blockages due to buildings. At the same time, it may also have many NLOS paths available due to refections by buildings and other structures. 
% \textcolor{red}{ Due to aforementioned reasons, indoor environment and cities are considered under urban scenario.} 
We evaluate our blockage analysis in the urban area with static blockages, LOS and NLOS paths along with dynamic and self-blockage. 
% \textcolor{red}{All of the indoor applications, V2X in cities, and M2M communication are considered in this scenario.}
%A network provider may have to support the QoS requirements of a verity of URLLC application in the metropolitan areas ranging from MTC, eHealth, tactile Internet, AR/VR, and autonomous driving.  
\end{enumerate}

The typical parameters used for simulation and numerical evaluation are presented in the Table~\ref{tab:params}. A list of applications as well as their latency and reliability requirements are presented in Table~\ref{tab:applications}. The checks and the cross marks in Table~\ref{tab:applications} represent whether caching can be used for the applications to satisfy the QoS requirements.

\subsection{Open Park Scenario}

\textbf{Simulation framework:} For open areas, we conducted our analysis using both the stochastic geometry model and the hexagonal cell deployment. 
% We first evaluate the effect of random topology on the probability and expected duration of blockages, and then present a comparison of the random model with the hexagonal cells case.   
We consider two values of dynamic blocker density, $0.01$ bl/m$^2$ and $0.1$ bl/m$^2$, and two values of the self-blocking angle $\omega$ ($0^{\text{o}}$ and $60^{\text{o}}$) for our study. We compare our analytical results using stochastic geometry with a MATLAB simulation\footnote{Our simulator MATLAB code is available at github.com/ishjain/mmWave.},
where the movement of blockers is generated using the random waypoint mobility model~\cite{randomWayPoint,randomWayPointSim}. For the simulation, we consider a square of size $200 \ \text{m}\times 200 \ \text{m}$ with blockers located uniformly in this area.
% that perfectly fit a disc of radius $R = 100$ m.
% Blockers are uniformly distributed in the considered rectangular area.
Our area of interest is the disc $B(o, R)$ of radius $R=100 $ m, which perfectly fits in the considered square area. The blockers choose a direction randomly, and move in that direction for a time-duration chosen uniformly in $[0,60]$ seconds. For the simulation, we performed 10,000 runs and each run consisted of the equivalent of 3 hours of blockers mobility. 
%\textcolor{red}{Each iteration in our simulation takes up to an hour of computation on CPU. We run the simulation for 10,000 iterations by parallelizing them on New York University High Performance Computing (HPC) clusters.} 
To maintain a fixed density of blockers in the square region, we consider that once a blocker reaches the edge of the square, it gets reflected. Note that for the given height of BSs, blockers, and the UE in Table \ref{tab:params}, we obtain from equation (\ref{eqn:rieff}) that the blockers can block the LOS link only when they are within a range corresponding to a small fraction (11\%) of the link length from the UE.

%We used the Mathwork code~\cite{randomWayPointSim} for this purpose. The simulation runs for an hour. We note the time instant when the blocker crosses a BS-UE link and generate a blockage duration through a realization of an exponential distribution with mean $\mu=2$. Further, we collect the time-series of alternate blocked/unblocked intervals for all the BS-UE links and take their intersection to obtain a time-series that represent the events of blockage of all available BSs. The blockage probability, frequency, and duration can be obtained from this time-series. Finally, we repeat the procedure for 10,000 iterations and report the average results. Rest of the simulation parameters are presented in Table~\ref{tab:params}.   

\begin{figure}[!t]
	\vspace{-2mm}
	\centering
	\includegraphics[width=0.43\textwidth]{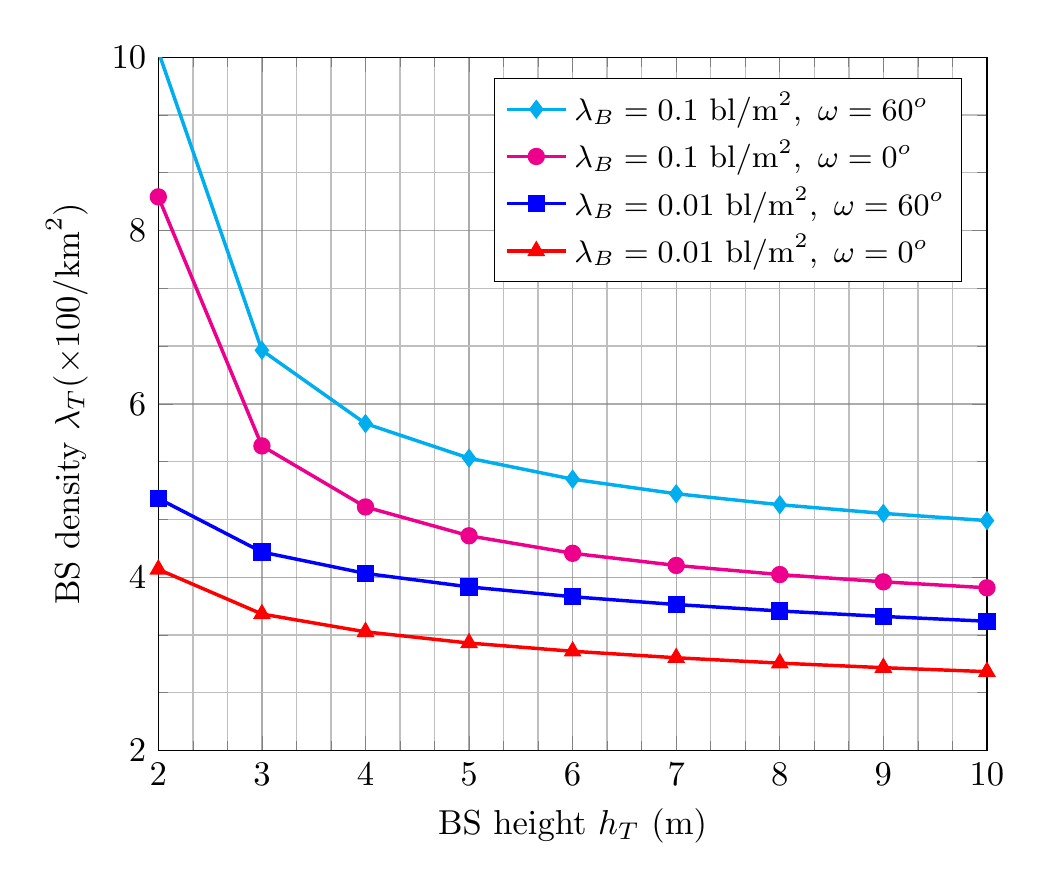}
    \captionsetup{skip=-2pt}
	\caption{Open park scenario: the trade-off between BS height and density for fixed dynamic blockage probability $P(B^{d}|\setC^{d})=10^{-5}$. }
	\label{fig:heightvsdensity}
	\vspace{-6mm}
\end{figure}

Figure~\ref{fig:condBl} presents the blockage statistics (blockage probability, expected blockage frequency, and expected blockage duration) of LOS links in the open park scenario for a communication range of 100 m. %\textcolor{red}{In this scenario, we consider only V2X and car entertainment services that include AR/VR.} 
For the analysis of reliability, we obtained the blockage probability and expected blockage duration, when the UE has at least one serving BS, i.e., the UE is always in the coverage area of at least one BS. 
The error bars in Figure~\ref{fig:condBl} represents a single standard deviation from the average value. 

\textbf{Impact of blocker density and communication range:} We consider the applications shown in Table~\ref{tab:applications} and evaluate the minimum density of BS required to satisfy their reliability and latency requirements. For V2X, the required reliability is 99\%, i.e., blockage probability 10$^{-2}$ and a maximum allowed latency of 20 ms. From Figure~\ref{fig:condBl}(\subref{fig:condBlprob}), for a communication range of 100 m, we can observe that to satisfy the reliability requirement of V2X, 200 BS/km$^2$ and less than 100 BS/km$^2$ may be required for blocker densities of 0.1 bl/m$^2$ and 0.01 bl/m$^2$, respectively. From Figure~\ref{fig:condBlR200}(\subref{fig:condBlprobR200}), we can observe that by increasing the communication range to 200 m, required number of BSs reduces significantly (less than 100 BS/km$^2$ and less than 50 BS/km$^2$ for blocker densities of 0.1 bl/m$^2$ and 0.01 bl/m$^2$, respectively). Furthermore, as caching is not a viable solution for V2X for achieving low latency, we can observe from Figure~\ref{fig:condBlR200}(\subref{fig:condBldurR200}) that approximately 200 BS/km$^2$ may be required to satisfy the latency requirement for $R=200$ m. However, a higher communication range (greater than 200 m) and traffic offloading employing the tight coupling of 5G NR and LTE network stacks may reduce the required BSs further.

 We now consider applications such as AR/VR car entertainment services, which require 99.99\% reliability and 20 ms maximum latency.
From Figure~\ref{fig:condBl}(\subref{fig:condBlprob}), for a communication range of 100 m, 
%we can observe that the blockage probability decreases exponentially with BS density.  For URLLC applications, this means that only a dense mmWave cellular network can achieve high reliability. 
we may require around 300 BS/km$^2$ and more than 400 BS/km$^2$ for blocker densities of 0.01 bl/m$^2$ and 0.1 bl/m$^2$, respectively. 
%decrease interruptions in the data transmission. For example, for a blocker density of 0.1 bl/m$^2$, a BS density of $100$/km$^2$ can decrease the interruptions to once in ten seconds, $200$/km$^2$ can decrease them to once in 100 seconds, and $300$/km$^2$
% \textcolor{red}{cellular radius is approximately $\frac{1}{\sqrt[]{\pi\lambda_T}}$}
%ecrease them to once in 1000 seconds. Reducing the frequency of interruptions is particularly crucial for AR/VR applications, therefore from this perspective a density of 200-300/km$^2$ may be required. This corresponds to about 6 to 9 BS, respectively,  within the range of each UE.
 %While this study presents a rough estimate of the required BS density for AR/VR streaming in car entertainment system, one can do a study for this by considering the highway lanes system and refined parameters.
From Figure~\ref{fig:condBlR200}(\subref{fig:condBlprobR200}), with a communication range of 200 m, the number of required BSs reduces significantly (less than 100 BS/km$^2$ and 150 BS/km$^2$ for blocker densities of 0.01 bl/m$^2$ and 0.1 bl/m$^2$, respectively). However, we can observe from Figure~\ref{fig:condBlR200}(\subref{fig:condBlfreqR200}) that the frequency of blockages also reduces significantly with a higher value of communication range. Thus, to mitigate the effect of these rare blockage events, caching of content at the UE can be used to achieve the reliability and latency requirements of AR/VR applications. We can observe from Figure~\ref{fig:condBlR200}(\subref{fig:condBldurR200}) that for a communication range of 200 m, caching of $40-60$ ms worth of data is required for a BS density $150$ BS/km$^2$ to have uninterrupted service. Thus, a BS density of $150$ BS/km$^2$ may be sufficient to achieve both reliability and latency of AR/VR for car entertainment services. As mentioned earlier, this number can be further reduced by considering a higher communication range (greater than 200 m) or tight coupling of 5G NR and LTE network stacks.

\textbf{Accuracy of simulation results:} From Figure~\ref{fig:condBl}(\subref{fig:condBlprob}) and~\ref{fig:condBl}(\subref{fig:condBlfreq}), we observe that both simulation and analytical results are approximately the same for a blocker density of $0.01$ bl/m$^2$, but deviates modestly for a high blocker density of $0.1$ bl/m$^2$, especially for high BS densities. Note that this deviation is due in part to our assumption that there is no more than one blocker blocking the link at the same time. However, no such assumption is made in simulations.
From Figure~\ref{fig:condBl}(\subref{fig:condBldur}), we observe our analytical results on blockage duration follow closely the simulation results for low blocker density ($0.01$ bl/m$^2$). For a high blocker density (0.1 bl/m$^2$), the percentage deviation is higher but still acceptable. 

\begin{figure}[!t]
\centering
\includegraphics[width=0.43\textwidth]{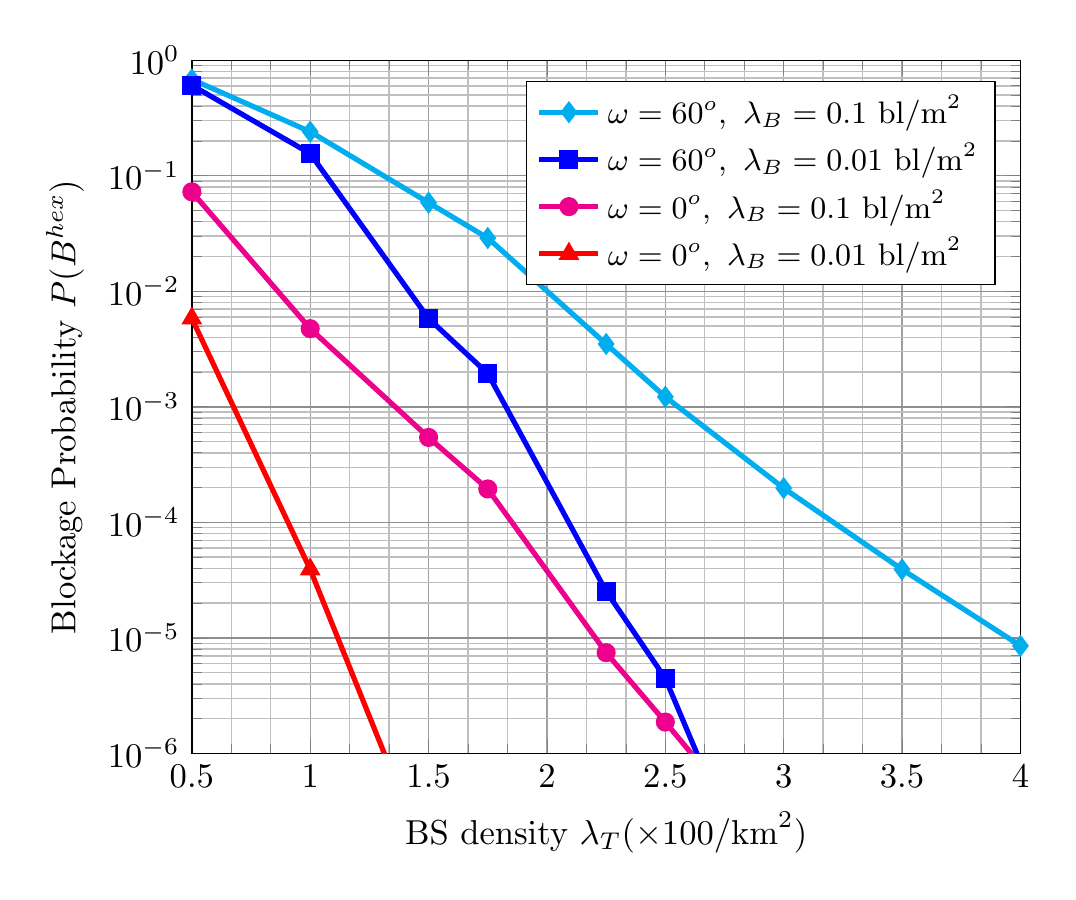}
\captionsetup{skip=-2pt}
\caption{Hexagonal cell open park scenario: blockage probability.}
\label{fig:pB_Hex}
\vspace{-6mm}
\end{figure}

\begin{figure}[!t]
    \centering
    \includegraphics[width=0.42\textwidth]{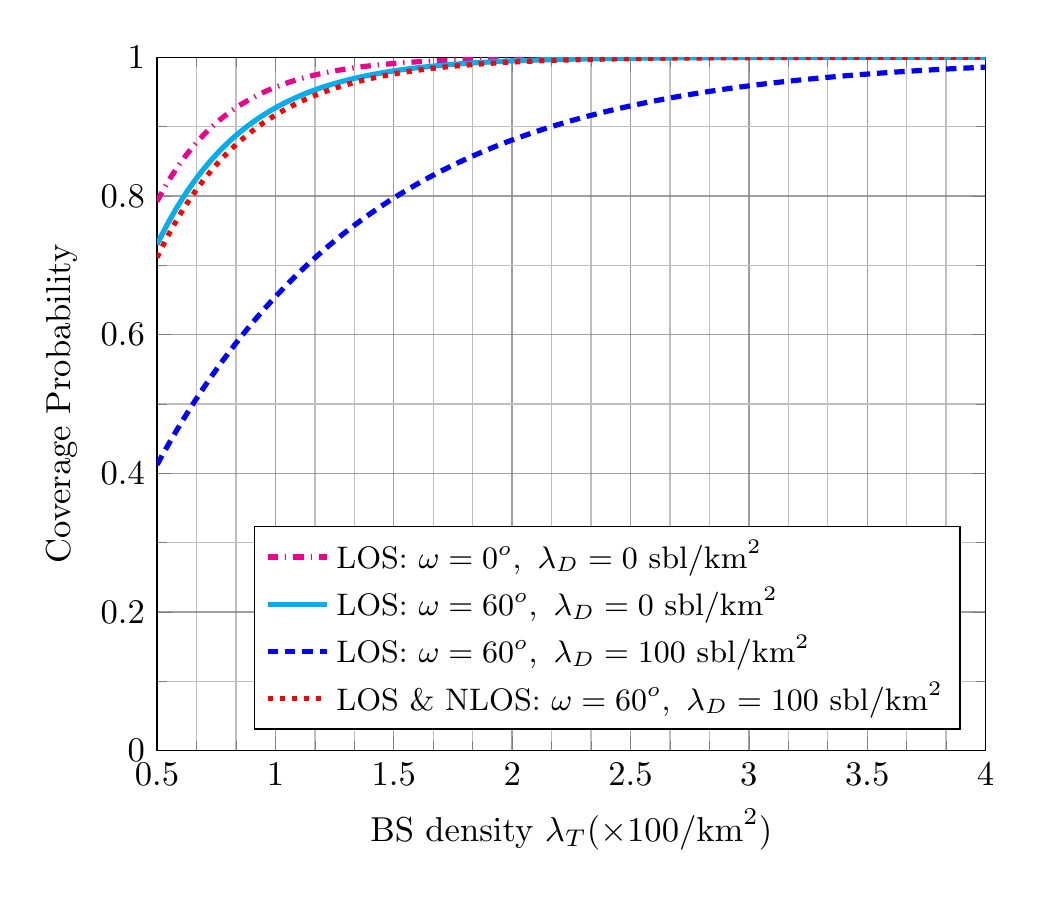}
    \captionsetup{skip=-2pt}
    \caption{Coverage probability for various scenarios.}
    \label{fig:coverage}
    \vspace{-5mm}
\end{figure}

\begin{figure*}[!t]
    \centering
    \begin{subfigure}[b]{0.4\textwidth}
        \captionsetup{skip=-2pt}    
        \includegraphics[width=\textwidth]{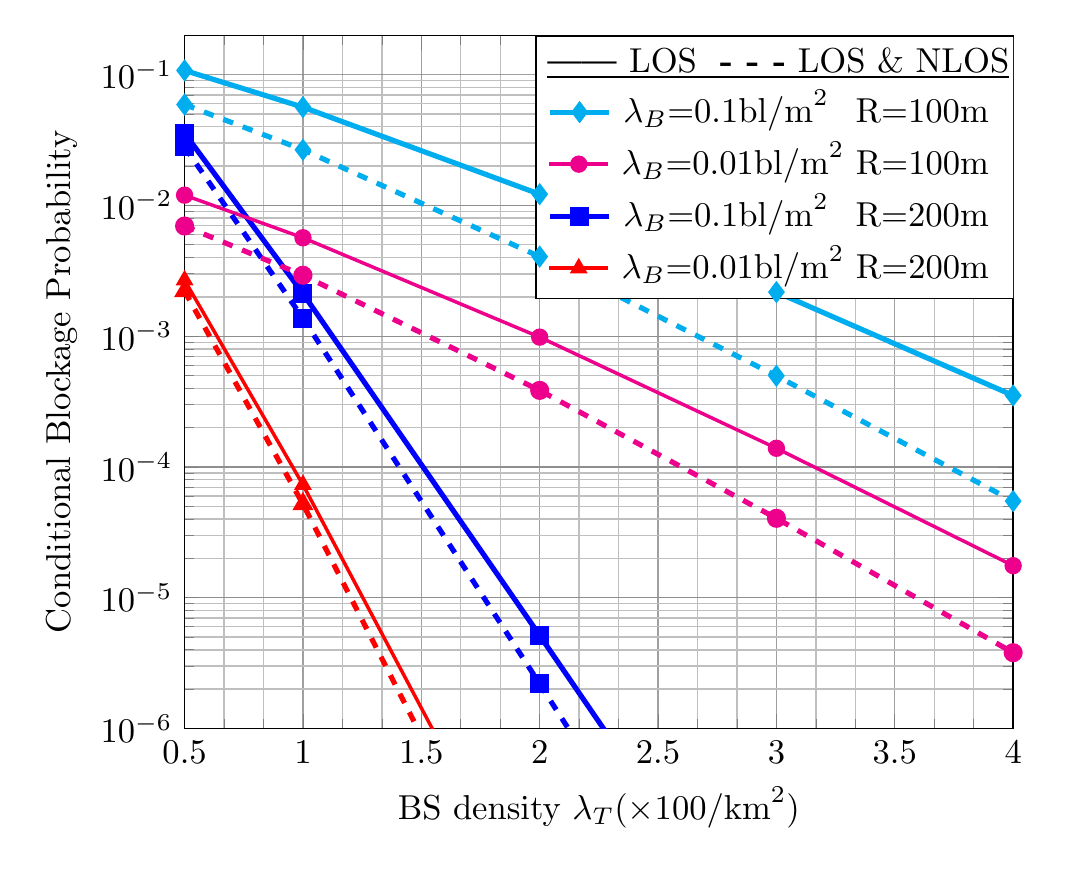}
        \caption{\small{Conditional blockage Probability (given coverage).}}
        \label{fig:StaticcondBlprob}
    \end{subfigure}\hspace{12mm}%
    \begin{subfigure}[b]{0.39\textwidth}
        \captionsetup{skip=-2pt}    
        \includegraphics[width=\textwidth]{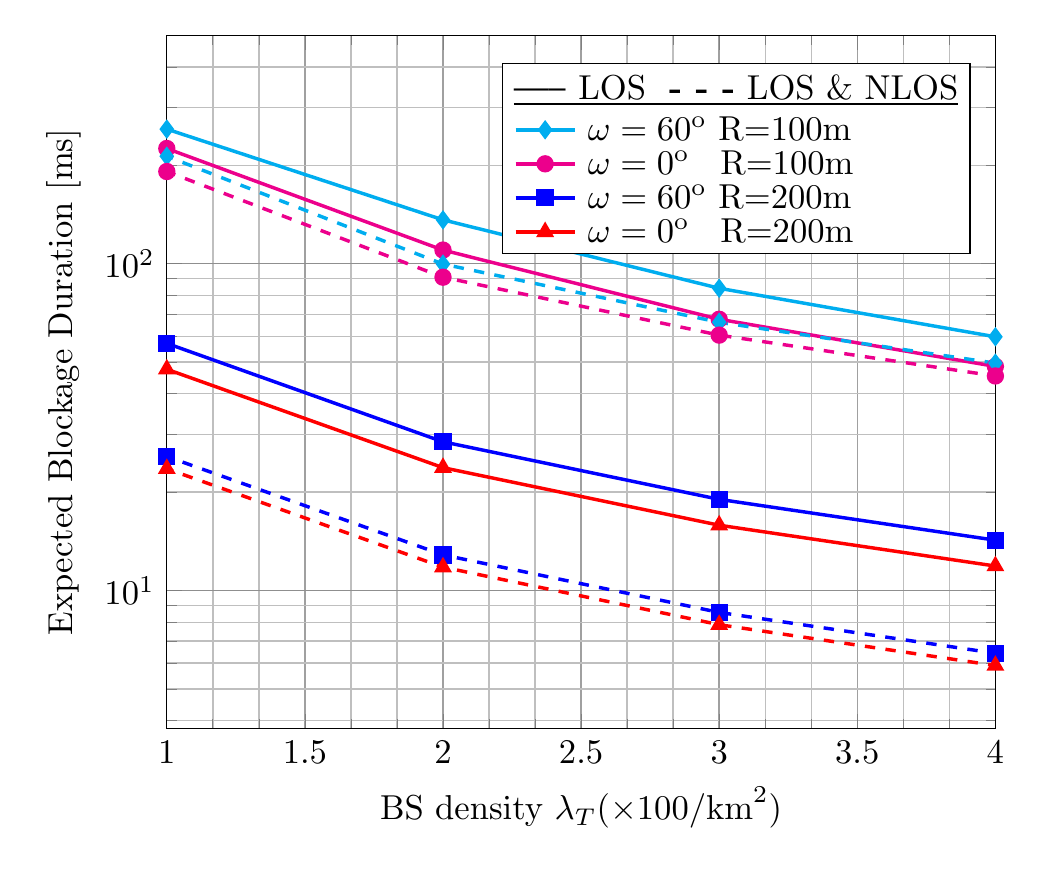}
        \caption{\small{Expected blockage duration (given coverage).}}
        \label{fig:StaticcondBldur}
    \end{subfigure}%
\caption{Urban scenario: Conditional blockage probability and duration (given coverage) considering 1) only LOS path, 2) both LOS and strong reflections (NLOS paths). We consider a static blockage density $\lambda_S=100$ sbl/km$^2$ and self-blockage angle $\omega=60^{\text{o}}$. Note that blockage duration is independent of dynamic blocker density.}
\label{fig:staticcondBl}
\vspace{-5mm}
\end{figure*}

\textbf{BS height--density tradeoff:} To increase the service reliability, apart from increasing the BS density, placing the BSs at a greater height may reduce the probability of blockage. The BS height vs. density trade-off is shown in Figure \ref{fig:heightvsdensity}. Note, for example, that doubling the height of the BS from 4 m to 8 m reduces the BS density requirement by approximately 20\%. 
% for blocker density $\lambda_B=0.1$  bl/m$^2$ and self-blockage angle $\omega = 60^{o}$.
The optimal BS height and density can be obtained by performing a cost analysis based on this trade-off.

\textbf{Results for the hexagonal cell model:} Finally, we present the results for the hexagonal cell model in Figure~\ref{fig:pB_Hex} and compare them with those for the random model in Figure \ref{fig:condBl}(\subref{fig:condBlprob}) for a communication range of 100 m. Note that for deterministic locations of BSs in hexagonal cells case, we were unable to get the closed-form solution and we used numerical integration to evaluate the performance. For the self-blockage angle 0$^{\text{o}}$, we observe that the blockage probability of $10^{-4}$  can be achieved with about half the BSs (less than 100 BS/km$^2$ for blocker density 0.01 bl/m$^2$, and less than 200 BS/km$^2$ for blocker density 0.1 bl/m$^2$). Next, we consider a self-blockage angle of $60^\text{o}$, where we computed an upper bound on the blockage probability. In the hexagonal cell case, with a blocker density of 0.01 bl/m$^2$ and self-blockage angle of 60$^{\text{o}}$, an upper bound of 200 BS/km$^2$ will be sufficient to achieve $10^{-4}$ blockage probability, which is significantly lower than that needed for the random topology in Figure \ref{fig:condBl}(\subref{fig:condBlprob}) (300 BS/km$^2$).

\textbf{Discussion on the rate of handovers:} To mitigate the effect of blockages by mobile blockers, the UE needs to handover frequently. Given that there is at least one unblocked BS (which occurs with a probability close to 1 for high BS densities), the rate of handover is equivalent to the blockage rate of a single BS-UE link (given in (\ref{eqn:SingleBS})). We obtain the average handover rate to be 0.05 handovers/s/UE and 0.5 handovers/s/UE for blocker densities of 0.01 bl/m$^2$ and 0.1 bl/m$^2$, respectively for $R=100$ m. These handover rates are much higher than those typically seen in 4G systems. This highlights the need for appropriate protocol design and resource allocation to make these  frequent handovers seamless so that they do not affect application layer QoS.

% We take the average BS-UE distance as $\frac{2R}{3}$ and obtain the average handover rate in table below: 
% }

% \begin{tabular}{|c|c|c|}\hline
%      Handover rate& $\lambda_B=0.01$ bl/m$^2$ &$\lambda_B=0.1$ bl/m$^2$\\\hline
%     $R=100$ m & 0.05 handovers/s & 0.5 handovers/s\\ \hline
%     $R=200$ m & 0.09 handovers/s&0.9 handovers/s\\ \hline
% \end{tabular}
% Although the handover process is very fast, a high frequency of handover may cause additional overhead to the system of The handover 

\subsection{Urban Scenario}
In an urban environment, the LOS paths to the UE may be blocked by buildings or permanent structures in addition to the dynamic and self-blockage. However, reflections from the buildings and other permanent structures may help in achieving higher coverage and service reliability by providing additional NLOS paths.
\textbf{Coverage analysis:} From Figure~\ref{fig:coverage}, we can observe that the coverage may degrade significantly due to static and self-blockage if the NLOS paths are not available. To achieve coverage of 90\% in the urban scenario, around 200 BS/km$^2$ may be required in the absence of NLOS paths. However, to achieve 90\% coverage in the presence of NLOS paths, significantly fewer BSs will be required ($<$ 100 BS/km$^2$). 
For NLOS links, we used the NLOS communication range of $\Rt = R\times 10^{-\frac{\gamma^{\text{NLOS}}}{10\times \textrm{PLE}}}$, where PLE is the path loss exponent and $\gamma^{\text{NLOS}}$ is the attenuation (in dB) due to reflection of the signal. We use PLE$=2.69$ given by the path loss model in~\cite{CellularCap-Rap} at 73 GHz and assume an average attenuation of $\gamma^{\text{NLOS}}=5$ dB. We thus get a NLOS range $\Rt=65$ m corresponding to the LOS range $R=100$ m and an $\Rt=130$ m corresponding to $R=200$ m. When considering these results, note that prior work on a capacity analysis in \cite{CellularCap-Rap} and \cite{bai2015coverage} suggest that a density of only 30-100 BS/km$^2$ is enough to meet the capacity requirements of mmWave users. 

\textbf{Impact of static blockages and NLOS paths:} Figure~\ref{fig:staticcondBl} presents the blockage probability and duration in the urban scenario. 
% for two values of the communication range ($R=100$ m and $R=200$ m). 
% We observe that the LOS blockage probability 
We observe that NLOS paths provide a surprisingly modest benefit in achieving a lower probability of blockage and expected blockage duration.
This is due to the fact that a UE may have far fewer BS within $\Rt$, which is significantly less than $R$, to provide NLOS paths as compared to LOS paths. UEs which are almost isolated, and are connected to only a few relatively distant BS, will suffer blockages disproportionately, and may have no NLOS paths to mitigate this effect.
% However, the probability of blockage in the NLOS case is bounded by the fact that there is a high probability of not having any real BS within $\Rt$, i.e., there is a high probability that the UE will not have any strong NLOS path. 
% On the other hand, a higher value of communication range can help significantly in achieving a low probability of blockage and expected blockage duration.

We now consider the applications mentioned in Table~\ref{tab:applications} and analyze the required number of BSs to satisfy both reliability and latency requirements. For applications such as smart grid and the tactile Internet, where high reliability (99.999\%, i.e., $10^{-5}$ blockage probability) is required, we may require a very high number of BSs, when we assume a communication range of 100 m. However, with a 200 m communication range, we may need less than 200 BS/km$^2$ (still a very high number) to achieve 99.999\% reliability (See Figure~\ref{fig:staticcondBl}(a)). Note that for V2X, tactile Internet, and smart grid, caching is not a viable solution. Thus, we may require around 200 BS/km$^2$ to satisfy the latency requirement of these applications (See Figure~\ref{fig:staticcondBl}(b) and Table~\ref{tab:applications}). The results suggest that even with a rich scattered environment and higher communication range, we need a potentially uneconomically high number of mmWave BSs to satisfy the QoS requirements of URLLC applications as compared to the BSs required for capacity and coverage~\cite{CellularCap-Rap,bai2015coverage}, which are typically of the order of $30-100$ BS/km$^2$. Thus, we may require tight coupling of different RATs such as 5G NR, LTE, and WiFi to collectively achieve QoS requirement for 5G applications to maintain a lower 5G mmWave BS density. Furthermore, note that most of these URLLC applications will require a $0$ ms SIT, further necessitating the tight coupling of different RATs. Note that the above discussion is necessarily tentative since mmWave 5G networks are only now being deployed. Experience with such network deployments, and the resultant technological improvements, may require us to revise our conclusions. However, we believe that the methodology developed in this paper, with appropriate amendments,  can still be used in designing future 5G mmWave networks.

\section{Conclusions}
\label{sec:conclusion}
In this paper, we presented simplified models to quantify key QoS parameters such as blockage probability and duration in mmWave cellular systems. We presented a generalized blockage analysis considering dynamic blockage due to mobile blockers, static blockages due to buildings and permanent structures and self-blockage due to the user's own body. The user is considered blocked when LOS and NLOS paths to BS around the UE are blocked simultaneously. We verified our theoretical model with MATLAB simulations for self-blockage and dynamic blockages. For the scenarios we considered, our results indicate that the density of BS required to provide an acceptable quality of experience for URLLC applications is much higher than that obtained by capacity or coverage requirements. These results suggest that the mmWave cellular network engineering may be driven by dynamic blockage rather than capacity or coverage requirements. Furthermore, the blockage events may be correlated for multiple BSs based on the blocker's size and location. This correlation, which we did not model, may result in an even higher blockage probability. We also present an analysis of blockage probability for regularly spaced hexagonal cells and showed that such a planned mmWave cellular architecture could reduce blockages events as compared to more randomly allocated BS locations. As pointed out in the introduction, sub-6 GHz bands could be used to maintain connectivity during dynamic blockages, but this requires tight control plane integration between the mmWave and sub-6 GHz bands and careful traffic engineering to prevent the random traffic overflow from mmWave bands from overwhelming sub-6 GHz capacity. In the future, we plan to address this issue, and issues related to correlation between blockage events. 

\appendix
\subsection{Proof of Lemma \ref{lemma:PN}}\label{app:proof_PNn}
The probability that a BS-UE link is not blocked (by static or self-blockage) follows from the independence of static blockages and the self-blockage by user's body. Denote $\mathcal{C}_i$ as the event that the $i$th BS is not blocked by either static blockage or self-blockage. We calculate the probability $P(\mathcal{C}_i|m)$ by utilizing the expressions for static and self-blockage from (\ref{eqn:pBstatic}) and (\ref{eqn:self}) respectively, and taking the average over the distance distribution $f_{R_i|M}(r|m)$ from (\ref{eqn:distribution}) as follows:
\begin{equation} \label{eqn:tempintegral}
\begin{split}
P(\mathcal{C}_i|m) &= \int_{r=0}^R pe^{-(\beta r+\beta_0)} \frac{2r}{R^2}dr\\
&=pq,
\end{split}
\end{equation}
where $q = \int_{r=0}^R e^{-(\beta r+\beta_0)} \frac{2r}{R^2}dr$ is solved to a closed-form expression given in (\ref{eqn:p}). We assume that given $m$ BSs in the disc $B(o,R)$, each BS may get blocked independently with probability $pq$. Therefore, the distribution of the number of BSs $n$ which are not blocked by static or self-blockage follows a binomial distribution,
\begin{equation}
P_{N|M}(n|m) = \binom{m}{n}(pq)^n(1-pq)^{m-n},\qquad n\le m.
\end{equation}
Taking the average over the distribution $P_M(m)$ given in (\ref{eqn:poisson}), we get the distribution of $N$ as follows:
\begin{equation}
\begin{split}
P_N(n) &= \sum_{m=0}^{\infty}P_{N|M}(n|m)P_M(m)\\
&= \sum_{m=n}^{\infty}\!\!\binom{m}{n} (pq)^n(1-pq)^{m-n} \frac{[\lambda_{T} \pi R^2]^m}{m!}e^{-\lambda_{T} \pi R^2}\\
&= \!\!\!\sum_{m-n=0}^\infty \!\!\!\frac{1}{(m-n)!}\left((1-pq)\lambda_T\pi R^2\right)^{m-n}\!\!e^{-(1-pq)\lambda_{T} \pi R^2}\\&\qquad\qquad\qquad\qquad\qquad\quad\times\frac{[pq\lambda_{T} \pi R^2]^n}{n!}e^{-pq\lambda_{T} \pi R^2} \\
&= \frac{[pq\lambda_{T} \pi R^2]^n}{n!}e^{-pq\lambda_{T} \pi R^2}.
\end{split}
\end{equation}
Note that we obtain the last equality using the fact that the sum of a Poisson distribution over its range $[0,\infty]$ is one.
This concludes the proof of Lemma \ref{lemma:PN}.

% \end{proof}

\subsection{Proof of Theorem 1}\label{app:th_prob}
The probability $P(B^{LOS}|m)$ is given by
\begin{equation}\label{eqn:proofpB1_firsthalf}
\begin{split}
&P(B^{LOS}|m)\\&=\!\!\int_{r_1}\!\!\!\cdots\!\int_{r_m} \prod_{i=1}^m \;P(B_i^{LOS}|m,r_i) f(\{r_i\}|m)\; dr_1\cdots dr_m\\
& =\!\!\int_{r_1}\!\!\!\cdots\!\int_{r_m} \prod_{i=1}^m\left[\left(1-pe^{-(\beta r_i+\beta_0)}\frac{1}{1+\frac{C}{\mu}r_i}\right) f(r_i|m)\right] \\&\qquad\qquad\qquad\qquad\qquad\qquad\qquad \times dr_1\cdots dr_m.\\
% & = \prod_{i=1}^m\int_{r=0}^R \left(1-pe^{-(\beta r+\beta_0)}\frac{1}{1+\frac{C}{\mu}r}\right)\frac{2r}{R^2}dr\\
\end{split}
\end{equation}
Note that the $m$-fold integral in (\ref{eqn:proofpB1_firsthalf}) can be solved separately (as the integrand can be separated into independent products) by solving $m$ identical integrals of the kind $\int_{r_i=0}^{R} \left(1-pe^{-(\beta r_i+\beta_0)}\frac{1}{1+\frac{C}{\mu}r_i}\right)\frac{2r_i}{R^2} d r_i$ as follows:
\begin{equation}\label{eqn:proofpB1}
\begin{split}
&P(B^{LOS}|m)\\
&=\left(\int_{r=0}^R \left(1-pe^{-(\beta r+\beta_0)}\frac{1}{1+\frac{C}{\mu}r}\right)\frac{2r}{R^2}dr\right)^m\\
&=\left(1-p\int_{r=0}^R \frac{e^{-(\beta r+\beta_0)}}{1+\frac{C}{\mu}r}\frac{2r}{R^2} dr\right)^m\\
&=(1-ap)^m,
\end{split}
\end{equation}

where we defined $a$ in (\ref{eqn:a}).
We now evaluate $P(B^{LOS})$ as
\begin{equation}\label{eqn:probPBpart2Dyn}
\begin{split}
   &P(B^{LOS})=\sum_{n=0}^{\infty} P(B^{LOS}|m)P_M(m)\\
   &=\sum_{m=0}^\infty (1-ap)^m  \frac{[\lambda_{T} \pi R^2]^m}{m!}e^{-\lambda_{T} \pi R^2} \\
 & = \sum_{m=0}^\infty \frac{[(1-ap)\lambda_{T} \pi R^2]^m}{m!}e^{-(1-ap)\lambda_{T} \pi R^2} e^{-ap\lambda_{T} \pi R^2}\\
&= e^{-ap\lambda_{T} \pi R^2}, \\
\end{split}
\end{equation}
where the last equality is obtained using the fact that the sum of a Poisson distribution over its range $[0,\infty]$ is one.
Finally, the conditional blockage probability $P(B^{LOS}|\setC^{LOS})$, conditioned on the coverage event $\setC^{LOS}$, is obtained as follows:
\begin{equation}
P(B^{LOS}) = P(B^{LOS}|\setC^{LOS})P(\setC^{LOS})+1-P(\setC^{LOS})
\end{equation}
and therefore,
\begin{equation}\label{eqn:pureBlk}
\begin{split}
P(B^{LOS}|\setC^{LOS})
&= \frac{P(B^{LOS})-\left(1-P(\setC^{LOS})\right)}{P(\setC^{LOS})}\\
&=\frac{e^{-a p\lambda_T\pi R^2} - e^{-pq\lambda_T\pi R^2}}{1-e^{-pq\lambda_T\pi R^2}}.
\end{split}
\end{equation}

This concludes the proof of Theorem \ref{th1}.

% \end{proof}

We now proceed with the proof of Corollary \ref{cor:LOSprobOpen}. We first derive $P(B^d|m)$ in a manner similar to (\ref{eqn:proofpB1}), but by setting $\beta=\beta_0=0$ for the dynamic blockage case without any static blockages (equivalently, we set $q=1$). Therefore, the expression of $a$ in (\ref{eqn:a}) gets simplified and is represented by $a'$ as follows:
\begin{equation}\label{eqn:aprimeDerive}
\begin{split}
a'&=\int_{r=0}^R \frac{1}{1+\frac{C}{\mu}r}\frac{2r}{R^2} dr\\
& = \frac{2\mu}{RC}-\frac{2\mu^2}{R^2C^2} \log\left(1+\frac{RC}{\mu}\right),
\end{split}
\end{equation}
where the intermediate steps of integration are omitted for brevity.
The rest of the analysis is similar to the derivation in (\ref{eqn:probPBpart2Dyn}) and (\ref{eqn:pureBlk}).

\subsection{Proof of Theorem \ref{th3}}\label{app:th_dur}
% \begin{proof}

Using the results from (\ref{eqn:exp_TBgivenN}), we find the expected blockage duration $\E\left[T^{LOS}|\setC^{LOS}\right]$, conditioned on the coverage event $\setC^{LOS}$ defined in (\ref{eqn:coveragep}), as follows:
\begin{equation}
\begin{split}
 &\mathbb{E} \left[T^{LOS}|\setC^{LOS}\right]\\
 &= \frac{\mathbb{E} \left[T^{LOS},n\ne 0\right]}{P(\setC^{LOS})}= \frac{\sum_{n=1}^\infty \frac{1}{n\mu} P_N(n)}{P(\setC^{LOS})} \\
&= \frac{\sum_{n=1}^\infty \frac{1}{n\mu} \frac{[pq\lambda_{T} \pi R^2]^n}{n!}e^{-pq\lambda_{T} \pi R^2}}{1-e^{-pq\lambda_T\pi R^2}} \\
&= \frac{e^{-pq\lambda_{T} \pi R^2}}{\mu\left(1-e^{-pq\lambda_T\pi R^2}\right)}\sum_{n=1}^\infty \frac{[pq\lambda_{T} \pi R^2]^n}{n n!}.
\end{split}
\end{equation}
This concludes the proof of Theorem \ref{th3}.

% \end{proof}
% \textcolor{red}{Is it important?}
% \textcolor{blue}{
% Let us consider the series expansion of $e^x$.
% \begin{equation}
% \begin{split}
% \MoveEqLeft e^x = 1+x+\frac{x^2}{2!}+\frac{x^3}{3!}+\frac{x^4}{4!}+\frac{x^5}{5!}+\cdots \qquad \qquad \qquad \\
% \MoveEqLeft e^x = 1+\sum_{n=1}^\infty \frac{x^n}{n!} \implies e^x-1 = \sum_{n=1}^\infty \frac{x^n}{n!} \\ \MoveEqLeft \implies \frac{e^x-1}{x} = \sum_{n=1}^\infty \frac{x^{n-1}}{n!}\\
% \end{split}
% \end{equation}
% Integrating both side, we have
% \begin{equation}
% \begin{split}
% \MoveEqLeft \int_{0}^{\lambda_T\pi R^2} \frac{e^x-1}{x} dx = \sum_{n=1}^\infty \int_{0}^{\lambda_T\pi R^2} \frac{x^{n-1}}{n!} dx\\
% \MoveEqLeft \text{E}\text{i}\left[\lambda_T\pi R^2\right] = \int_{0}^{\lambda_T\pi R^2} \frac{e^x-1}{x} dx = \sum_{n=1}^\infty \frac{[\lambda_{T} \pi R^2]^n}{n n!}.
% \end{split}
% \end{equation}
% Hence,
% \begin{equation*}
% \mathbb{E}\left[T_B|\setC\right] = \frac{e^{-\lambda_T\pi R^2}}{\mu\left(1-e^{-\lambda_T\pi R^2}\right)}\text{E}\text{i}\left[\lambda_T\pi R^2\right].
% \end{equation*}
% }%!!!!!!!!!!!!!!!!!!!!!!!!!!!!!imp!!!!!!!

\subsection{Approximation of the expected duration}\label{app:approx_dur}
The expectation of a function $f(n) = 1/n$ can be approximated using the Taylor expansions for the moments of functions of random variables~\cite{benaroya2005probability} as follows:
\begin{equation}\label{eqn:approxf_n}
\begin{split}
\E[f(n)] &= \E[f(\mu_n+(n-\mu_n))],\\
&\approx\E[f(\mu_n)+f'(\mu_n)(n-\mu_n)+\frac{1}{2}f''(\mu_n)(n-\mu_n)^2]\\
&= f(\mu_n)+\frac{1}{2}f''(\mu_n)\sigma_n^2=\frac{1}{\mu_n}+\frac{\sigma_n^2}{\mu_n^3},
\end{split}
\end{equation}
where $\mu_n$ and $\sigma_n^2$ are the mean and variance of Poisson random variable $N$ given in (\ref{eqn:PN}). We get the required expression by substituting $\mu_n = pq\lambda_T\pi R^2$ and $\sigma_n^2 = pq\lambda_T\pi R^2$ in (\ref{eqn:approxf_n}) for $f(n)=1/n$ as follows:
\begin{equation}
\E[1/n] \approx \frac{1}{pq\lambda_T\pi R^2}+\frac{1}{(pq\lambda_T\pi R^2)^2}.
\end{equation}
On further simplification for high BS densities, we can further approximate the expression as:
\begin{equation}\label{eqn:approxE1byn}
\E[1/n] \approx \frac{1}{pq\lambda_T\pi R^2} = \frac{1}{\E[n]}.
\end{equation}
Using (\ref{eqn:approxE1byn}), we approximate $\E\left[T^{LOS}|\setC^{LOS}\right]$ as follows:
\begin{equation}\label{eqn:durLOS_approx}
\begin{split}
\left[T^{LOS}|\setC^{LOS}\right] 
% &=  \E\left[\frac{1}{n\mu}|n\ne 0 \right]\\
 &\approx \frac{1}{P(n\ne0)\mu\E[n,n\ne0]}\\
&= \frac{1}{P(\setC^{LOS})\mu pq\lambda_T\pi R^2}.
\end{split}
\end{equation}
Finally, the expression in (\ref{eqn:durLOSapproxfinal}) can be obtained by substituting $P(\setC^{LOS})$ from (\ref{eqn:coveragep}) to (\ref{eqn:durLOS_approx}).

\subsection{Proof of Theorem \ref{th2}}\label{app:th_freq}
% \begin{proof}
Note that the probability $P(B_i^d|n,\{r_i\})$ follows the same expression as $P(B_i^d|m,\{r_i\})$ in (\ref{eqn:AllBSdynamic}) in the absence of static blockages. Therefore, we simplify $\zeta^d$ in (\ref{eqn:AllBSfreq}) as follows:
\begin{equation}
\zeta^d = n\mu  \prod_{i=1}^n \frac{\frac{C}{\mu}r_i}{1+\frac{C}{\mu}r_i}.
\end{equation}
We first evaluate $\E[\zeta^d|n]$ by following the steps similar to the derivation of (\ref{eqn:aprimeDerive}) as follows:
\begin{equation}
\begin{split}
\E[\zeta^d|n] &= n\mu \left(\int_r  \frac{\frac{C}{\mu}r}{1+\frac{C}{\mu}r}  \frac{2r}{R^2} dr\right)^n \\
& = n\mu (1-a')^n,
\end{split} 
\end{equation}
where $a'=1-\int_r  \frac{\frac{C}{\mu}r}{1+\frac{C}{\mu}r}  \frac{2r}{R^2} dr$ was solved to the closed-form expression in (\ref{eqn:aprime}). Next, we evaluate $\E[\zeta^d]$ as follows:
% where the first expression follow from~(\ref{eqn:blockageDurEQn1}) and the second expression follows from~(\ref{eqn:AllBSfreq}) and~(\ref{eqn:AllBS}). Finally, the third expression is evaluated using independence of distance distribution. Following from~(\ref{ep1n})
\begin{equation}\label{eqn:derive_Ezeta_d}
\begin{split}
 \mathbb{E} \left[\zeta^d\right]  &=  \sum_{n=0}^\infty  n \mu(1-a')^n \frac{[p\lambda_{T} \pi R^2]^n}{n!}e^{-p\lambda_{T} \pi R^2} \\
 &= \sum_{n=1}^\infty  \frac{[(1-a')p\lambda_{T} \pi R^2]^{(n-1)}}{(n-1)!}e^{-(1-a')p\lambda_{T} \pi R^2}\\ &\qquad\qquad\qquad\quad\times\mu(1-a')p\lambda_{T} \pi R^2e^{-ap\lambda_{T} \pi R^2}\\
 &= \mu(1-a')p\lambda_{T} \pi R^2e^{-a'p\lambda_{T} \pi R^2}.\\
\end{split}
\end{equation}

Finally, the expected frequency of blockage $\E[\zeta^d|\setC^d]$ conditioned on the coverage $C^d$ (\ref{eq:coverage_dyn}) is given as follows:
\begin{equation}
\begin{split}
 \mathbb{E} \left[\zeta^d|\setC^d\right] &= \frac{\sum_{n=1}^\infty \E[\zeta^d|n] P_N(n)}{P(\setC^d)} =\frac{\sum_{n=0}^\infty \E[\zeta^d|n] P_N(n)}{P(\setC^d)} \\
 &= \frac{\E[\zeta^d]}{P(\setC^d)}=\frac{\mu(1- a')p\lambda_T\pi R^2e^{-a'p\lambda_T\pi R^2}}{{1-e^{-p\lambda_T\pi R^2}}}.
\end{split}
\end{equation} 
This concludes the proof of Theorem~\ref{th2}. 

% \end{proof}

\subsection{Proof of Lemma \ref{lemma:NLOScoverage}}\label{app:proof_NLOS_coverage}
% \begin{proof}

Assuming the independence of LOS and NLOS links, we obtain the coverage probability $P(\setC)$ as follows:
\begin{equation}
\begin{split}
P(\setC) &= \sum_{m}\!\!\int_{r_1}\!\!\!\cdots\!\int_{r_m}\!\!\!P(\setC|m,r_i)f(\{r_i\}|m)dr_1\cdots dr_mP_M(m)\\
&=\sum_{m}\!\!\int_{r_1}\!\!\!\cdots\!\int_{r_m}\!\!\!\left(1-\prod_{i=1}^m\left(1\!-\!pe^{-(\beta r_i+\beta_0)}\right)\!\!\left(1\!-\!I_{(r_i\le \Rt)}\right)\!\!\right)\\
	&\qquad\qquad\qquad\qquad \times f(\{r_i\}|m)dr_1\cdots dr_mP_M(m)\\
&=1-\sum_m \left(\int_{r=\Rt}^{R}\left(1-pe^{-(\beta r+\beta_0)}\right)\frac{2r}{R^2}dr\right)^mP_M(m)\\
&=1-\sum_m \left(1-\qt\right)^mP_M(m),  
\end{split}
\end{equation}
where we defined $\qt$ as
\begin{equation}\label{eqn:qtIntegral}
\qt = 1-\int_{r=\Rt}^{R}\left(1-pe^{-(\beta r+\beta_0)}\right)\frac{2r}{R^2}dr.
\end{equation}
We solve the integration in (\ref{eqn:qtIntegral}) to a closed-form expression given in (\ref{eqn:qt}). Finally, we evaluate the coverage probability (\ref{eqn:coverageNLOS}) in Lemma \ref{lemma:NLOScoverage} by following the steps similar to the derivation of (\ref{eqn:probPBpart2Dyn}).

\bibliographystyle{IEEEtran}
\bibliography{references}

% \subsection{Static Blockage Probability} \label{app:th_prob_static}

% Using the independence of the distribution of $R_i$, we get the final expression as
% \begin{equation}
% \begin{split}
% P(B_s|N) &= \int_rP(B_s|R_i,N) f(R_i)\;dr\\
% &=\prod_i^n \int_r(1-e^{-(\beta+\beta_0)r})\frac{2r}{R^2}\;dr\\
% &= \left(1-\frac{2}{\beta^2R^2}\left(e^{-\beta_0}-e^{-(\beta R+\beta_0)}(1+\beta R)\right)\right)^n\\
% &=(1-b)^n\\
% \end{split}
% \end{equation}
% Now, $P(B_s)$ can be obtained similar to (\ref{eqn:probPBpart2Dyn}).

% \subsection{Proof of Theorem \ref{th:NLOS}}\label{app:th_prob_NLOS}
% \textcolor{red}{ Incomplete??}

\begin{IEEEbiography}[{\includegraphics[width=1in,height=1.25in,clip,keepaspectratio]{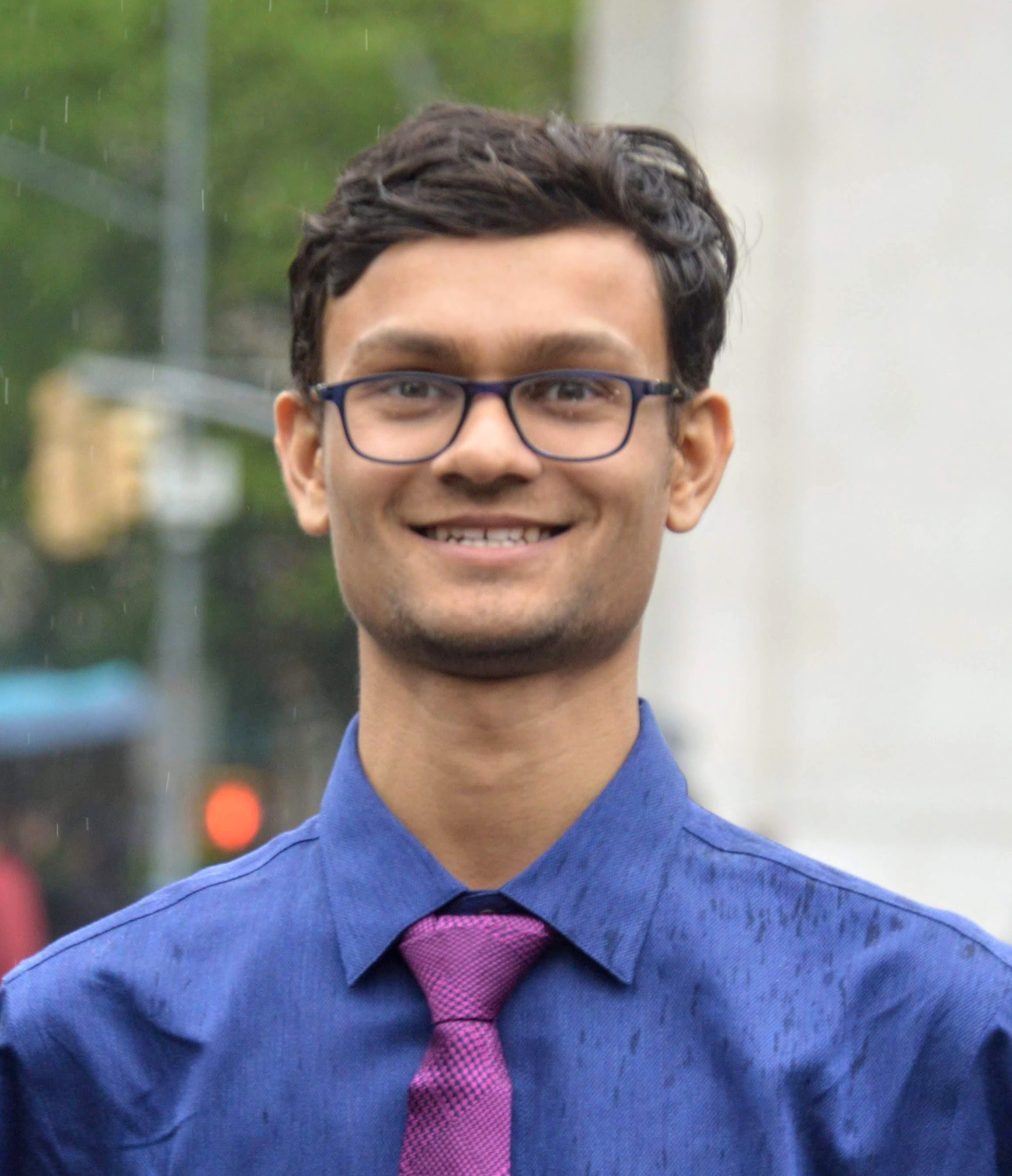}}]{Ish Kumar Jain}
received the B.Tech. degree from Indian Institute of Technology Kanpur, India, in 2016,
% He received the prestigious Motorola gold medal for the best all-round performance in electrical engineering during his B.Tech. convocation. 
and an MS degree from the Tandon School of Engineering, New  York University, NY, USA in 2018, both in Electrical Engineering. He was awarded the Samuel Morse Fellowship to pursue a Masters degree at NYU. He is currently pursuing a Ph.D. in Electrical Engineering at the University of California, San Diego, USA. His research interests include Wireless Communication Theory and Systems.
\end{IEEEbiography}

\begin{IEEEbiography}[{\includegraphics[width=1in,height=1in,clip,keepaspectratio]{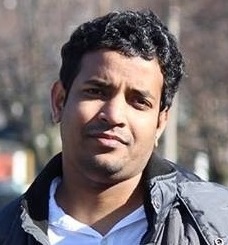}}]{Rajeev Kumar}
received the B.Tech. and M.Tech. degrees in Electrical Engineering from Indian Institute of Technology Madras, Chennai, India, in 2013. He is currently pursuing the Ph.D. degree in Electrical Engineering at the Tandon School of Engineering, New York University, NY, USA. He worked at Nokia Bell Labs during the summers of 2017 and 2018.
His research interests focus on latency issues related to 5G cellular systems.
\end{IEEEbiography}

\begin{IEEEbiography}[{\includegraphics[width=1in,height=1.25in,clip,keepaspectratio]{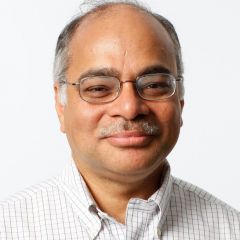}}]{Shivendra S. Panwar} (S'82-M'85-SM'00-F'11)
 is a Professor in the Electrical and Computer Engineering Department at the NYU Tandon School of Engineering. He received a Ph.D. degree in electrical and computer engineering from the University of Massachusetts, Amherst, in 1986. He is the Director of the New York State Center for Advanced Technology in Telecommunications (CATT), the Faculty Director and co-founder of the New York City Media Lab, and a member of NYU Wireless. His research interests include the performance analysis and design of networks. Current work includes cooperative wireless networks, switch performance and multimedia transport over networks. He has co-authored a textbook: ``TCP/IP Essentials: A Lab based Approach", Cambridge University Press. He was a winner of the IEEE Communication Society's Leonard Abraham Prize for 2004. He has served as the Secretary of the Technical Affairs Council of the IEEE Communications Society.
\end{IEEEbiography}

\end{document}